\documentclass[conference]{IEEEtran}

\ifCLASSINFOpdf
\else
\fi
\usepackage[style=acmnumeric, sorting=none]{biblatex}
\usepackage{framed}
\usepackage{boxedminipage}
\usepackage{graphicx}
\usepackage{xcolor}
\usepackage{empheq}
\usepackage{algorithm}
\usepackage{algorithmicx}  
\usepackage{algpseudocode} 

\AtBeginDocument{%
  }

\newcommand{\mathcolorbox}[2]{\colorbox{#1}{$\displaystyle #2$}}

\usepackage{color}
\usepackage{tabu}
\usepackage{colortbl}
\makeatletter
\newcommand*\bigcdot{\mathpalette\bigcdot@{.5}}
\newcommand*\bigcdot@[2]{\mathbin{\vcenter{\hbox{\scalebox{#2}{$\m@th#1\bullet$}}}}}
\makeatother
\usepackage{amsmath}
\usepackage{booktabs}
\usepackage{threeparttable}
\usepackage{tikz}
\usepackage{array,tabularx}
\usepackage[most]{tcolorbox}
\usepackage{PTSansNarrow}
\usepackage{amssymb}
\usepackage{amsthm}
\usepackage{subfigure}
\usepackage{enumerate}
\usepackage{enumitem}
\newtheorem{theorem}{Theorem}
\newtheorem{lemma}{Lemma}

\newtheorem{definition}{Definition}
\usepackage[breaklinks]{hyperref}
\usepackage{lipsum} 
\usepackage{multicol}
\usepackage{cases}
\usepackage{empheq}

\usetikzlibrary{matrix}

\begin{document}

%
\title{OptChain: Achieving Optimal Throughput of Permissionless Blockchains}

\author{
\IEEEauthorblockN{
Chunjiang Che\IEEEauthorrefmark{1}\IEEEauthorrefmark{2},
Songze Li\IEEEauthorrefmark{2}\IEEEauthorrefmark{3},
Xuechao Wang\IEEEauthorrefmark{1},
}
\IEEEauthorblockA{
\IEEEauthorrefmark{1}
The Hong Kong University of Science and Technology (Guangzhou)
}
\IEEEauthorblockA{
\IEEEauthorrefmark{2}Southeast University
}

\IEEEauthorblockA{
\IEEEauthorrefmark{3}
Engineering Research Center of Blockchain Application, Supervision and Management (Southeast University), \\
Ministry of Education
}
}

\maketitle

\begin{abstract}
We introduce \textit{OptChain}, a permissionless blockchain state machine replication (SMR) protocol that achieves optimal throughput. We first establish a theoretical upper bound on the throughput of any SMR protocol under a fixed error probability, and OptChain is the first protocol to approach this limit. Conceptually, OptChain is a sharding protocol that optimizes both vertical and horizontal scalability. Vertically, we introduce \textit{Shardis}, a novel permissionless verifiable information dispersal mechanism that maximizes intra-shard throughput to its physical limit, determined by the fastest node's bandwidth within each shard. Horizontally, we propose \textit{diffusion mining}, which ensures security as long as each shard includes at least one honest node, thereby allowing for the maximum number of shards. We provide a formal security and efficiency analysis, demonstrating that OptChain approaches the established upper bound while maintaining robust security. Finally, we implement a full prototype of OptChain and deploy it on AWS EC2 nodes across various regions. Experimental results indicate that OptChain outperforms state-of-the-art permissionless protocols and closely approaches the theoretical optimal throughput.
\end{abstract}


%
\IEEEpeerreviewmaketitle

\section{Introduction}
State machine replication (SMR) protocols, including longest chain-based approaches such as Proof of Work (PoW) and Proof of Stake (PoS), as well as voting-based approaches such as Byzantine Fault Tolerance (BFT), have attracted significant attention since 2008, when Satoshi Nakamoto introduced the concept of blockchain~\cite{nakamoto2008bitcoin}. These protocols enable a network of nodes to reach agreement on a single, consistent state in the presence of Byzantine faults. However, existing SMR protocols, including Nakamoto consensus~\cite{nakamoto2008bitcoin} and Practical Byzantine Fault Tolerance (PBFT)~\cite{castro1999practical}, inherently suffer from poor scalability, as their throughput is tightly constrained to ensure security. This tension is commonly referred to as the \textit{throughput-security trade-off}. Numerous efforts have sought to improve this trade-off, including vertically scaling solutions (e.g., Bitcoin-NG~\cite{Bitcoin-NG}, Prism~\cite{Prism}, HotStuff~\cite{HotStuff}, DispersedLedger~\cite{Dispersedledger}) and horizontally scaling solutions (e.g., sharding protocols~\cite{Elastico, Omniledger, Rapidchain, Manifoldchain, Monoxide}). Despite these advances, existing approaches still fall short of achieving optimal throughput while maintaining strong security guarantees.

\subsection{Optimal Throughput-Security trade-off}

We consider a system of $n$ nodes with heterogeneous bandwidths, where a fraction $\alpha > \frac{1}{2}$ are honest. Let $\overline{C}$ and $\underline{C}$ denote the maximum and minimum node bandwidths, respectively, measured in transactions per second. We establish an upper bound on the optimal throughput-security trade-off: for any SMR protocol that ensures an error probability less than $\sigma$, the throughput cannot exceed $T=\frac{\overline{C}}{1 - \sigma^{1/{\alpha n}}}$. The formal theorem is presented in Section~\ref{optimal_tps}.

This upper bound reveals two primary challenges in designing a SMR protocol to approach the theoretical optimum:

\begin{enumerate}
    \item When $\sigma = 0$, $T = \overline{C}$, indicating that the network must process transactions at the bandwidth of the fastest node.
    \item When $\sigma > 0$, $T > \overline{C}$, implying that the system must parallelize transaction processing—typically through a \textit{sharding} protocol. Moreover, the number of shards should be maximized (see analysis in Section~\ref{optimal_tps}).
\end{enumerate}

To the best of our knowledge, no SMR protocol can simultaneously address these two challenges in a permissionless network. Prism~\cite{Prism} is the closest permissionless protocol to address the first challenge, achieving a throughput of $0.9\underline{C}$; however, it remains limited by the slowest node in the network. While DispersedLedger~\cite{Dispersedledger} can achieve a throughput approaching $\overline{C}$, it is a permissioned protocol whose design is non-trivial to adapt to a permissionless setting. Furthermore, both are non-sharding protocols and inherently fail to address the second challenge. Conversely, Monoxide~\cite{Monoxide} and Manifoldchain~\cite{Manifoldchain} are the only sharding protocols that support maximal shards, yet they suffer from Bitcoin-like throughput within each shard and fail to address the first challenge. Consequently, no existing solution approaches this optimal throughput.

\subsection{Our Approach}

We present \textit{OptChain}, the first layer-1 SMR protocol that approaches the optimal throughput-security trade-off in a permissionless network. Fundamentally, OptChain is a sharding protocol that simultaneously maximizes both vertical and horizontal scalability, thereby inherently approaching the optimal throughput-security trade-off. The two challenges discussed above also represent the key barriers to maximize vertical and horizontal scalability, respectively.

Vertical scalability measures the improvement in throughput when computational and communication resources increase. When nodes in the network are configured with faster GPUs or higher bandwidths, the overall throughput should be increased. However, the throughput of most SMR protocols is limited by the slowest node with the lowest bandwidth. In a heterogeneous network with \textit{stragglers}—nodes with limited bandwidth—their vertical scalability is significantly hindered. To enhance vertical scalability, we propose a novel \textit{permissionless verifiable information dispersal} (PVID) protocol, \textit{shardis}, enabling throughput to scale with the fastest nodes despite stragglers. Intuitively, nodes first agree on an ordered log of \textit{commitments}, each serving as a compact digest (e.g., a Merkle root) of an available block. Subsequently, nodes download the full block data to update their state machines. This design ensures that state replication is driven by the fastest honest nodes, making throughput a function of the fastest rather than the slowest participants. Consequently, this approach effectively addresses Challenge-(1) and serves as a cornerstone for achieving the optimal throughput-security trade-off.

Horizontal scalability quantifies the improvement in throughput when the number of nodes increase. Sharding protocols have emerged as a promising approach to enhance horizontal scalability by dividing nodes into $S$ distinct shards, each maintaining an independent ledger. However, most existing designs~\cite{Elastico, Omniledger, Rapidchain} require each shard to hold an honest majority—typically at least $1/2$ or $2/3$ of the shard’s nodes. By the law of large numbers, shard sizes must remain sufficiently large to keep the probability of adversarial-majority shards low. As a result, the number of shards (S) is inherently limited, which ultimately constrains horizontal scalability. We introduce the \textit{diffusion mining} mechanism to ensure security as long as each shard contains at least one honest node, thereby maximizing $S$ and the horizontal scalability. Intuitively, diffusion mining allows honest nodes to diffuse their hashing power across shards by mining \textit{global availability blocks} that can be appended to chains in all shards. Even if some shards exhibit adversarial majorities (referred to as \textit{corrupted shards}), honest hashing power can be aggregated from other shards with honest majorities to secure them, ensuring the effective honest hashing power exceeds $50\%$. This approach effectively addresses Challenge-(2), serving as another cornerstone for approaching the optimal throughput-security trade-off.

The above design introduces a new challenge in ensuring \textit{cross-shard atomicity}, which requires that coins deposits occur only after all corresponding withdrawals succeed. Existing sharding protocols~\cite{Manifoldchain, Monoxide} adopt the \textit{Two-Phase Commit} (2PC) protocol~\cite{2pc} to commit cross-shard transactions. In this approach, nodes use Merkle proofs to verify whether coins have been successfully locked in the withdrawal shards. The presence of a cross-shard transaction in the longest chain indicates that it has been validated by honest nodes in that shard, enabling the creation of coins on the deposit shards. However, OptChain cannot directly apply this mechanism, as it verifies transactions only after retrieving full blocks. This deferred verification introduces two challenges: (1) invalid withdrawals may be exploited to generate false deposits, and (2) network delays may cause some honest nodes to record a deposit before the corresponding withdrawal in their confirmed ledgers. To address these issues, OptChain employs an \textit{ordering chain} to globally sequence availability blocks, ensuring deposits occur only after all corresponding withdrawals, and uses \textit{Fraud Proofs}~\cite{fraudproof} to sanitize the ledger by removing invalid transactions.

We implement a prototype of OptChain in $7,000$ lines of Rust code~\cite{optchain_anon} and evaluate it on a geo-distributed AWS EC2 testbed across four regions. To ensure fidelity, we utilized real-world bandwidth traces~\cite{realbandwidth} to capture network heterogeneity. We extensively evaluated OptChain against the SOTA baselines—Manifoldchain and Prism—as well as an empirically derived theoretical optimum. Our results demonstrate that OptChain's throughput closely approaches the theoretical optimum, significantly outperforming existing protocols across varying error thresholds. Moreover, further experiments confirm that OptChain achieves both optimal vertical and horizontal scalability, thereby significantly substantiating our aforementioned claims.

\subsection{Main Contribution}

We highlight our contributions as follows:

\begin{enumerate}
    \item We establish an upper bound on throughput that constrains any permissionless SMR protocol.
    \item We propose shardis, a novel PVID scheme that outperforms existing solutions, providing a pathway to optimal vertical scalability in permissionless systems.
    \item We propose OptChain, a permissionless sharding protocol that approaches optimal throughput, and we formally prove both its security and throughput optimality.
    \item We implement a full system prototype of OptChain and evaluate it on AWS EC2 instances spanning four regions, showing that its throughput experimentally approaches the theoretical optimum.
\end{enumerate}

\section{Related Work}

Many works have sought to enhance blockchain scalability, either vertically or horizontally, yet none have achieved optimality in both dimensions.

\noindent {\bf Vertically Scaling Solutions.} Among existing solutions~\cite{Bitcoin-NG, Prism, HotStuff, Dispersedledger}, Prism stands out as the most vertically scalable permissionless protocol. Its throughput approaches the physical limit, the network’s communication capacity, scaling linearly with the bandwidth of the slowest node. Prism decomposes the blockchain into three fundamental functionalities: leader election, transaction proposal, and ancestor voting. Accordingly, a full block is divided into three distinct components, proposer blocks, transaction blocks, and voter blocks, each corresponding to one of these functions. This decoupling enables fast consensus on lightweight proposer blocks and allows each to reference arbitrary number of transaction blocks, scaling throughput toward the network’s communication capacity. However, a significant gap remains between Prism’s throughput and the optimal, as it is limited by the slowest node’s bandwidth, not to mention its zero horizontal scalability.  

\noindent {\bf Horizontally Scaling Solutions.} Sharding protocols are a fundamental approach to achieving horizontal scalability. Most existing sharding protocols~\cite{Elastico, Omniledger, Rapidchain} require each shard to maintain an honest majority, limiting the number of shards and hindering horizontal scalability. Manifoldchain~\cite{Manifoldchain} addresses this limitation by tolerating corrupted shards, ensuring cross-shard security as long as each shard contains at least one honest node. It introduces two block types: exclusive blocks, which function like regular blocks, and inclusive blocks, which can be appended to all shard chains. Inclusive blocks allow honest nodes to distribute their hashing power across shards, reinforcing the security of those shards with weaker honest hashing power. However, while Manifoldchain maximizes horizontal scalability, its throughput within each shard remains comparable to Bitcoin, reflecting poor vertical scalability. Motivated by this limitation, we leverage a similar methodology to design diffusion mining, distinguishing our approach by allowing shard throughput to approach the bandwidth of the fastest node.

\section{Preliminaries}

\subsection{Scalable SMR via Sharding}\label{sharding_smr}

An SMR protocol, as its name suggests, enables multiple machines to agree on the same state despite Byzantine behavior. In this context, a client–server model is typically assumed: We consider a system of $n$ servers, each maintaining a replica of the state machine, where up to a fraction $\beta = 1- \alpha$ of them are malicious. Clients submit transactions to all servers to update or read the machine state. The servers then execute an SMR protocol to agree on a consistent, totally ordered log of transactions, which they subsequently execute to maintain the replicated state. For simplicity, we use the term \textit{node} instead of \textit{server}, as both share the same meaning but node is more common in the blockchain context.

Traditional SMR protocols suffer from poor scalability in large-scale networks, as every node must replicate the entire machine state, resulting in significant overlap in communication and computation overhead. The blockchain sharding protocol is proposed to address this problem. It partitions the nodes into multiple shards, each responsible for executing a distinct subset of transactions. Specifically, both transactions and nodes are assigned shard identifiers (IDs), and a node only executes transactions that share its shard ID. This design enables different nodes to work in parallel, allowing the overall system throughput to scale proportionally with the number of shards—and inherently, with the number of nodes. Formally, a sharding protocol provides the following properties, which can be categorized as intra-shard and cross-shard security properties:

\begin{definition}[Intra-shard Security Property]
\quad
\begin{itemize}
    \item \textbf{Safety:} If two honest nodes in the same shard execute sequences of transactions $\{\mathtt{tx}_1, \ldots, \mathtt{tx}_j\}$ and $\{\mathtt{tx}'_1, \ldots, \mathtt{tx}'_{j'}\}$, then $\mathtt{tx}_i = \mathtt{tx}'_i$ for all $i \leq \min\{j, j'\}$.
    \item \textbf{Liveness:} If an honest client sends a transaction $\mathtt{tx}$ to all nodes at time $t$, then $\mathtt{tx}$ will eventually be executed by all honest nodes in a specific shard by time $t+u$.
\end{itemize}
\end{definition}

In sharding protocols, a cross-shard transaction (\textit{cross-tx}) consisting of multiple inputs (withdrawals) and outputs (deposits) may span different shards, and the protocol must also guarantee:

\begin{definition}[Cross-shard Security Property]
\label{cross_shard_atomicity_definition}
\quad
\begin{itemize}
    \item \textbf{Cross-shard Atomicity:} For any cross-tx distributed across withdrawal and deposit shards: (1) if all withdrawal operations are confirmed, all deposit operations will eventually be confirmed; (2) if any deposit is confirmed, all withdrawals must already be confirmed and all deposits will eventually be confirmed; and (3) if initiated by an honest user, the entire cross-tx (including withdrawals and deposits) will eventually be confirmed.
\end{itemize}
\end{definition}


\subsection{Verifiable Information Dispersal}\label{prelim_VID}

Verifiable Information Dispersal (VID) is an effective approach for reducing communication overhead in blockchains, and has been adopted by many permissioned protocols~\cite{Dispersedledger, simplex}. Specifically, it allows a node to run a \textit{dispersal} protocol to distribute blocks across various nodes, and later execute a \textit{retrieval} protocol to reconstruct the original block. Through this process, honest nodes can agree on a block's availability by downloading only a block fragment. We present the formal definition of VID below, consistent with prior works~\cite{Dispersedledger, VID}.

\begin{definition}
A Verifiable Information Dispersal scheme consists of three core procedures:
\begin{itemize}
    \item \textsc{Disperse}($B$) $\rightarrow$ $\mathtt{com}$: A deterministic algorithm executed by the proposer of a block $B \in \{0,1\}^*$ to initiate the dispersal process and returns a block commitment $\mathtt{com}$.
    \item \textsc{Verify}($\mathtt{com}$) $\rightarrow$ $\{0,1\}$: A randomized interactive protocol invoked by validators to determine whether the dispersal of the block associated with commitment $\mathtt{com}$ has completed. It outputs $1$ if the dispersal is complete, and $0$ otherwise.
    \item \textsc{Retrieve}($\mathtt{com}$) $\rightarrow B$: An interactive protocol invoked by a node to recover a data block $B \in \{0,1\}^*$.
\end{itemize}
\end{definition}

Formally, a VID protocol is secure if it satisfies the following properties except with negligible probability.

\begin{itemize}\label{VID_security_properties}
    \item \textbf{Termination:} If an honest node invokes \textsc{Disperse}($B$), it eventually gets an output $\mathtt{com}$, and all honest nodes eventually output \textsc{Verify}($\mathtt{com}$) $= 1$.
    
    \item \textbf{Agreement:} If any honest node outputs \textsc{Verify}($\mathtt{com}$) $= 1$, then every honest node will do the same within a network delay of $\Delta$.
    
    \item \textbf{Retrievability:} If an honest node outputs \textsc{Verify}($\mathtt{com}$) $= 1$, then any honest node that invokes \textsc{Retrieve}($\mathtt{com}$) eventually reconstructs some block $B'$.
    
    \item \textbf{Correctness:} For any two successful \textsc{Retrieve}($\mathtt{com}$) calls that yield $B_1$ and $B_2$, where \textsc{Verify}($\mathtt{com}$) $= 1$, it holds that $B_1 = B_2$. Furthermore, if $\mathtt{com}$ was generated by an honest node via \textsc{Disperse}($B$), then $B_1 = B$.
\end{itemize}

VID is an effective technique for improving throughput independent of slow nodes. DispersedLedger~\cite{Dispersedledger} employs VID to scale throughput with the fastest node’s bandwidth. Specifically, block producers disperse block symbols via VID and broadcast a commitment to all nodes, which quickly reach consensus on an ordered set of commitments. Each node then retrieves and reconstructs the corresponding blocks at its own pace, allowing fast nodes to advance the ledger more quickly.

\noindent {\bf Permissionless VID.} Although many VID schemes have been proposed for permissioned protocols~\cite{Dispersedledger}, deploying one that operates effectively in a permissionless setting remains challenging. In a permissioned network, each node must authenticate itself to join the protocol, ensuring that all nodes know the total number of participants \( n \) and the upper bound on corrupted nodes \( \beta n \). One can add redundancy to a block using erasure coding, which can be recovered from any $e$ symbols, where $e$ is a constant. Thus, if $\beta n + e$ nodes report holding a symbol, dispersal is complete, as at least $e$ honest nodes must have received it. However, in a permissionless system, the parameter $n$ is unknown, because an adversary can distribute its hashing power across multiple devices to arbitrarily increase the number of corrupted nodes. 

Fisch \textit{et al.} proposes the first permissionless VID scheme~\cite{PermissionlessVID}. Intuitively, nodes vote ``yes'' or ``no'' on a block commitment via the underlying SMR protocol, based on whether they have received the requested symbols. A block is considered successfully dispersed if the difference between ``yes'' and ``no'' votes exceeds a threshold. However, this permissionless VID design loses gain as the adversarial ratio $\alpha$ approaches $1/2$. Specifically, the scheme requires any $(\alpha-\beta)n$ nodes to reconstruct the block, implying each node must store many symbols. Denoting the block size as $b$ bytes, the per-node communication cost is $O\left(\frac{b}{(\alpha-\beta)n}\right)$, which degrades to $O(b)$ as $(\alpha-\beta)n \rightarrow 1$. In this case, each node must download the full block to complete the dispersal. Furthermore, it requires each node to request an equal number of symbols, resulting in uniform workload distribution across all nodes for a block. This uniformity makes such schemes unsuitable for sharding protocols, where in-shard nodes should request more symbols than out-shard nodes because they eventually need to retrieve the full block---thus, requesting more symbols does not impose additional overhead on them. In this paper, we propose shardis, which achieves a per-node communication cost of $O\!\left(\frac{b}{\alpha n}\right)$, as detailed in Appendix~\ref{communication_complexity}. In our design, the workload of each block is primarily handled by in-shard nodes, while out-shard nodes are responsible for only a negligible fraction.


\subsection{Coded Merkle Tree}\label{coded_merkle_tree}

As previously mentioned, our protocol ensures security even if each shard contains only a single honest node. The Coded Merkle Tree (CMT)~\cite{CMT} is a crucial building block for this, because it provides constant-cost protection against data availability attacks, even when the majority of nodes are malicious. Specifically, CMT employs an \((vd, d)\) erasure code to introduce redundancy: a block \(B\) of \(b\) bytes is evenly divided into \(d\) data symbols, \(B = [m_1, \dots, m_d]\), each of size \(\frac{b}{d}\) bytes. These symbols are linearly combined to produce a coded block \(CB = [c_1, \dots, c_{vd}]\), where \(v\) denotes the coding rate. The original block \(B\) can be reconstructed from any \(d\) of the \(vd\) coded symbols. When constructing the Merkle tree, the second layer is built on the coded block’s symbols and re-encoded with the same erasure code to form the third layer. This process repeats iteratively until the final root is obtained.

Any honest node can verify full data availability by downloading only a block commitment of size $O(1)$ bytes and randomly sampling $O(\log b)$ bytes. Specifically, given a block $B$, CMT encodes the block and computes its commitment. If every honest node requests $O(\log b)$ bytes upon receiving the commitment, the following properties hold:
\begin{itemize}
    \item \textbf{Soundness:} If an honest node determines that a block is available, then at least one honest node will be able to recover the block within a constant delay. 
    \item \textbf{Agreement:} If an honest node determines that a block is available, then all other honest nodes in the system will also determine that the block is available within a constant delay.
\end{itemize}



\section{Throughput-Security Trade-off}
In this section, we introduce a formal model for analyzing security and throughput, and establish a fundamental upper bound on the throughput of any permissionless SMR protocol.

\subsection{Security Model}
We summarize our assumptions below, which are consistent with those commonly adopted in permissionless blockchain protocols~\cite{nakamoto2008bitcoin,Prism,Monoxide,Manifoldchain}.

\begin{itemize}
    \item The system comprises $n$ nodes, denoted by $\{P_1, \dots, P_n\}$, operating in a synchronous and permissionless network.
    \item A weakly adaptive adversary controls a fraction $\beta < 1/2$ of the nodes, while the remaining fraction $\alpha = 1-\beta$ are honest. Corruption is dynamic but requires a delay to effectuate.
    \item Nodes have heterogeneous bandwidths, where $P_i$ has a download capacity of $b_i$. We denote the maximum and minimum bandwidths as $\overline{C}$ and $\underline{C}$, respectively.
    \item For simplicity, we assume all nodes possess homogeneous computational power.
\end{itemize}

\subsection{Optimal Throughput}\label{optimal_tps}

We derive the optimal throughput under these assumptions above. We also provide the intuition behind this upper bound, while the formal proof is deferred to Appendix~\ref{optimal_throughput_security_trade-off}.

\begin{theorem}\label{throughput_security_tradeoff}
If an SMR protocol guarantees that the probability of any confirmed transaction violating at least one of the security properties is bounded by $\sigma$, then the upper bound on throughput is $T^* = \frac{\overline{C}}{1 - {\sigma}^{\frac{1}{\alpha n}}}$.

\end{theorem}

\noindent \textbf{Proof Sketch.} First, we relax the assumption on nodes' bandwidths and allow each node to download transactions at the same rate as the fastest node. In other words, each node can download at most $\overline{C}$ transactions per second, so the entire network can download at most $n\overline{C}$ transactions per second. Then, each transaction appears, on average, in at most $\frac{n\overline{C}}{T}$ copies across the network, with each copy downloaded by one node. For any given transaction, the probability that a particular node downloads a copy is $\frac{\overline{C}}{T}$. To ensure availability, at least one honest node must download each transaction; otherwise, malicious nodes could withhold it, rendering the transaction unavailable and thus ineligible to be counted in throughput. Consequently, the probability that all copies of a transaction are downloaded only by corrupted nodes is $\left(1 - \frac{\overline{C}}{T} \right)^{\alpha n}$. By backward reasoning, we can derive the optimal throughput $\frac{\overline{C}}{1 - {\sigma}^{\frac{1}{\alpha n}}}$ given an error probability of $\sigma$.

Based on the proof sketch, achieving optimal throughput requires addressing two fundamental challenges:

\begin{enumerate}
    \item The protocol must be capable of processing transactions at the rate of the fastest node. As the error threshold $\sigma \to 0$, the throughput $T$ approaches $\overline{C}$, implying that the system must fully utilize the maximum available node bandwidth.
    
    \item The protocol must be a sharding protocol that ensures security as long as each shard contains at least one honest node. If the protocol requires more than one honest node per shard to ensure security, the failure condition expands (i.e., the error occurs even when one honest node is present). This effectively raises the error probability above the theoretical lower bound $\sigma'(T) > \left(1 - \frac{\overline{C}}{T} \right)^{\alpha n}$. Consequently, to maintain the target error bound $\sigma'(T) < \sigma$, the system would be forced to throttle throughput such that $T < T^*$.
\end{enumerate}

The analysis motivates the design of OptChain. Vertically, we design shardis to enable intra-shard consensus to tightly follow the fastest node's bandwidth. Horizontally, we introduce diffusion mining to ensure global security with a minimum of one honest node per shard.

\section{OptChain}





OptChain approaches the optimal throughput by simultaneously maximizing vertical and horizontal scalability. To better illustrate this, we start by showing how to maximize vertical scalability, i.e., optimizing the performance of a single chain, then extend the design to multiple shard chains to demonstrate how optimal horizontal scalability can be achieved. Intuitively, OptChain maximizes vertical scalability by implementing a permissionless VID protocol. Nodes first agree on an ordered log of block commitments, and then independently download the corresponding transactions to update their state machines at their own pace. 
Once each shard's throughput is optimized, OptChain employs diffusion mining to design a multi-chain structure, supporting an optimal number of shards, thereby maximizing horizontal scalability.

The full structure of OptChain is shown in Fig.~\ref{fig:full_structure}. OptChain consists of three distinct chains, each serving a specific function. First, nodes mine a \textit{proposer chain}, where each \textit{proposer block} contains commitments of multiple transaction blocks. Second, nodes verify the availability of these commitments by mining $S$ \textit{availability chains}. Each \textit{availability block} in the $i$-th chain includes references to commitments that (1) are confirmed in the proposer chain, and (2) correspond to transaction blocks belonging to shard~$i$. Finally, nodes mine an \textit{ordering chain} that references confirmed availability blocks across shards, establishing a global order for the availability blocks across shards. 


\begin{figure}
  \centering
  \includegraphics[width=8.5cm]{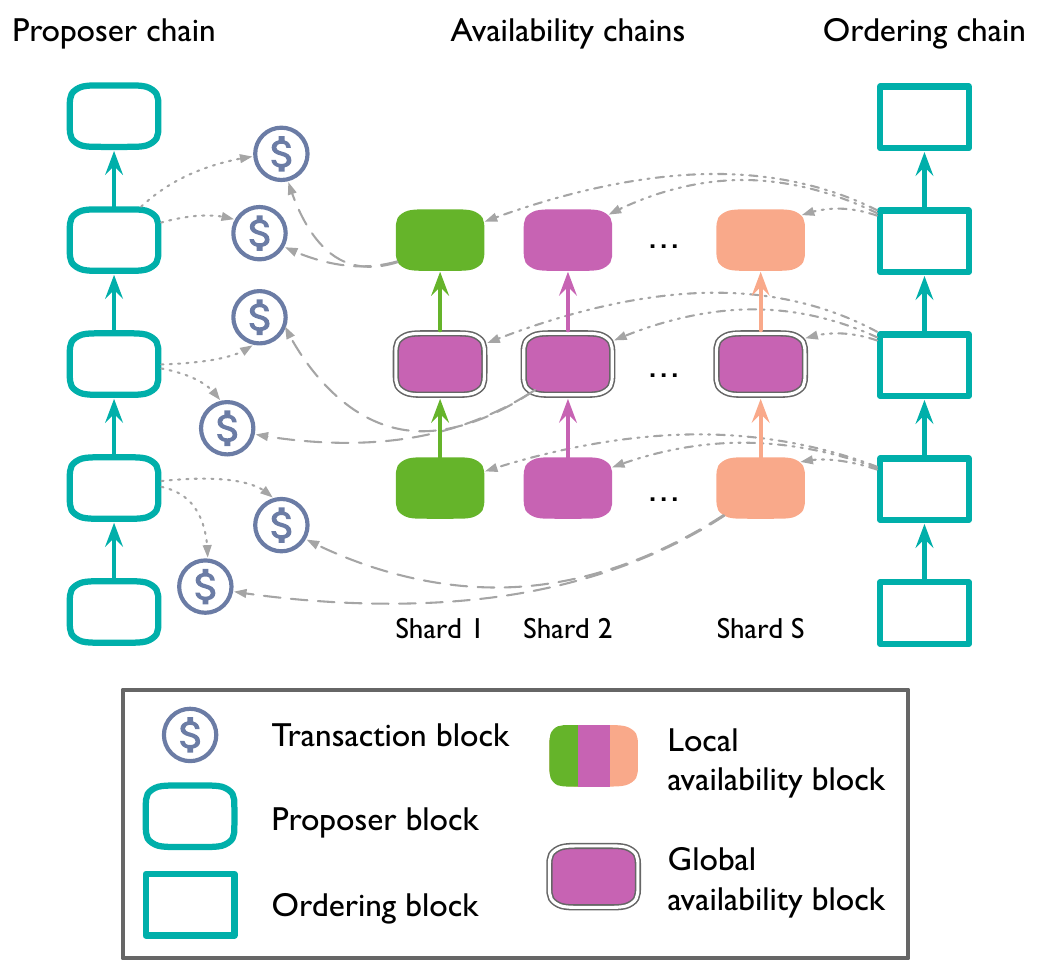}
  \vspace{-2mm}
  \caption{The structure of OptChain, which comprises five block types. Proposer blocks form a global beacon chain (left) that references multiple transaction blocks by including their commitments. Availability chains (middle) span $S$ shards, where availability blocks reference transaction blocks to vote on their availability. While local availability blocks extend a single chain, global availability blocks extend all availability chains simultaneously. Finally, the ordering chain (right) consists of ordering blocks that sequence the availability blocks to impose a global order.}
  \label{fig:full_structure}
  \vspace{-6mm}\hspace{-10mm}
\end{figure}

\subsection{Block and Chain Structures}\label{block_structure}

In OptChain, blocks are classified as proposer, transaction, availability, and ordering blocks. Each transaction block \( B_t \) is assigned a shard ID. It contains \( t \) transactions in the corresponding shard and a root computed by CMT, which serves as its commitment $\mathtt{com}$. 

\begin{itemize}
    \item \textbf{Transaction block $B_t$}:
    \begin{enumerate}
        \item $\mathtt{shard\_ID}$: an index denoting the specific shard to which the block is affiliated.
        \item $\mathtt{data\_blob}$: a set of $t$ transactions.
        \item $\mathtt{com}$: the root of a CMT constructed from the included transactions.
        \item $\mathtt{nonce}$: the resolved solution to a PoW puzzle.
    \end{enumerate}
\end{itemize}

A proposer block $B_p$ functions like a standard Bitcoin block but includes references to transaction blocks—their commitments—instead of the transactions themselves. The proposer has no shard ID, and the proposer chain serves as a global beacon chain:

\begin{itemize}
    \item \textbf{Proposer block $B_p$}:
    \begin{enumerate}
        \item $\mathtt{prop\_parent}$: the hash of the parent proposer block.
        \item $\mathtt{prop\_com\_set}$: a set of $\mathtt{com}$ of the transaction blocks.
        \item $\mathtt{nonce}$: the resolved solution to a PoW puzzle.
    \end{enumerate}
\end{itemize}

Availability blocks are assigned shard IDs and serve as votes for the availability of transaction blocks referenced by the proposer chain. Each availability block may include multiple block commitments confirmed in the proposer chain, indicating that the producer has verified their availability. Based on their hash, availability blocks are categorized as either \textit{local} or \textit{global availability blocks}. Local availability blocks extend a single chain, whereas global availability blocks can extend multiple chains simultaneously.

\begin{itemize}
    \item \textbf{Availability block $B_a$}:
    \begin{enumerate}
        \item $\mathtt{shard\_ID}$: an index denoting the specific shard to which the block is affiliated.
        \item $\mathtt{inter\_parent}$: the hash of the parent availability block within the affiliated shard.
        \item $\mathtt{global\_parents}$: a set composed of all parents across all shards.
        \item $\mathtt{avai\_com\_set}$: a set of $\mathtt{com}$ of the transaction blocks.
        \item $\mathtt{nonce}$: the resolved solution to a PoW puzzle.
    \end{enumerate}
\end{itemize}

An ordering block references multiple availability blocks, forming a chain that imposes a global order on confirmed availability blocks across shards.

\begin{itemize}
    \item \textbf{Ordering Block $B_o$}:
    \begin{enumerate}
        \item $\mathtt{order\_parent}$: the hash of the parent ordering block.
        \item $\mathtt{avai\_blk\_set}$: a set of hash of the availability blocks referenced by this block.
        \item $\mathtt{nonce}$: the resolved solution to a PoW puzzle.
    \end{enumerate}
\end{itemize}

\subsection{Shardis}\label{permissonless_vid_maintext}

Our permissionless VID scheme is designed based on the following intuition: The block producer broadcasts $\mathtt{com}$ when it initiates \textsc{Disperse}($B$) $\rightarrow$ $\mathtt{com}$. Nodes randomly request symbols upon receiving $\mathtt{com}$ and verify availability once all requested symbols are received. Nodes then run a longest chain-based protocol to agree on an ordered list of available $\mathtt{com}$. By the protocol's safety, every honest node will have the requested symbols for any confirmed $\mathtt{com}$. With an estimated $\alpha n$ honest nodes and a properly chosen symbol threshold per node, the full block can be reconstructed except with a negligible probability.

A naive implementation of this high-level idea can be achieved by slightly modifying an existing protocol, such as Prism. Unlike the nodes in the original Prism protocol, who accept a block only after receiving it in full, the nodes in our initial protocol accept a block upon receiving all the requested symbols of that block. However, this initial protocol faces a significant challenge. The adversary can propose an unavailable block, releasing only partial symbols, and respond solely to nodes requesting them. This can deceive a portion of honest nodes who request the released symbols, leading them to accept the block, while others reject it, causing a split view. Furthermore, the adversary can conceal any number of unavailable blocks and propose them simultaneously, causing an arbitrary number of honest nodes to split their views.

To resolve the above challenges, we consolidate our permissionless VID scheme by applying the following high-level insights:
\begin{itemize}
    \item \textbf{Using erasure coding to reduce the proportion of deceived nodes.} Erasure coding adds redundancy to the block, requiring the adversary to withhold more symbols to make the block unavailable. This reduces the number of nodes who are deceived by requesting the released symbols. For example, consider a block with $4$ symbols, where each honest node requests only one symbol. The adversary releases three symbols and withholds one, causing $3/4$ of honest nodes to accept this unavailable block. However, after encoding the block with a $(8, 4)$ erasure code, the adversary must withhold $5$ symbols since the full block can be retrieved with any $\geq 4$ symbols. This reduces the proportion of deceived nodes to $3/8$. 
    \item \textbf{Requesting more symbols for old blocks.} Nodes first mine a proposer chain to reach consensus on an ordered set of potential $\mathtt{com}$(s), and then request symbols for confirmed $\mathtt{com}$(s) based on their depth in the proposer chain. The deeper a $\mathtt{com}$ is in the proposer chain, the more symbols are requested. This approach prevents the adversary from deceiving an arbitrary number of honest nodes by hiding many blocks, as they must first post blocks in the proposer chain. Additionally, any long-withheld blocks require more symbols to deceive honest nodes, and the probability of success diminishes with the withholding time.
\end{itemize}


To realize the above insights, we 
decouple the proposer chain into two separate chains: a proposer chain and an availability chain, which handle the \textsc{Disperse}($B$) and \textsc{Verify}($\mathtt{com}$) procedures, respectively. Nodes may invoke \textsc{Retrieve}($\mathtt{com}$) at any time after dispersal is complete. 

\noindent {\bf Proposer chain (\textsc{Dispersal}).} Nodes treat instances of \textsc{Disperse}($B_t$) as transactions and execute the following protocol to append them to the proposer chain. Once the dispersal of an instance is complete, the invoker outputs a $\mathtt{com}$.

\begin{itemize}
    \item \textbf{Packing:} The block producer broadcasts the $\mathtt{com}$ of $B_t$ to all nodes. Each node collects all unreferenced $\mathtt{com}$(s) to populate the $\mathtt{prop\_com\_set}$ of a proposer block and selects the hash of the highest proposer block as its $\mathtt{prop\_parent}$.
    \item \textbf{Mining:} The node iterates over $\mathtt{nonce}$ values to solve the PoW puzzle and mines a valid proposer block, which is then broadcast to the network.
    \item \textbf{Verification:} A node accepts a proposer block if the included PoW solution is valid.
    \item \textbf{Confirmation:} A proposer block $B_t$ is confirmed once it is followed by $k$ proposer blocks. At that point, the node outputs \textsc{Disperse}($B_t$) $=$ $\mathtt{com}$.
\end{itemize}

\noindent {\bf Availability chain (\textsc{Verify}).} Nodes mine the availability chain to maintain the status of dispersed $\mathtt{com}$, serving as an oracle to respond to invokers of \textsc{Verify}($\mathtt{com}$).

\begin{itemize}
    \item \textbf{Sampling:} For any confirmed $\mathtt{com}$ at level $l$ in the proposer chain, and given the current highest level $l'$ of the chain, a node randomly requests $l' - l$ symbols of the block associated with that $\mathtt{com}$. 
    
    \item \textbf{Packing:} Each node gathers all unreferenced $\mathtt{com}$(s) for which the requested symbols have been received. These are used to populate the $\mathtt{avai\_com\_set}$ of an availability block. The node also selects the hash of the highest availability block as the $\mathtt{inter\_parent}$.
    
    \item \textbf{Mining:} The node searches for a valid $\mathtt{nonce}$ to solve the PoW puzzle and mines an availability block, which is then broadcast to the network.
    
    \item \textbf{Verification:} A node accepts an availability block if (1) the PoW solution is valid and (2) it has received all the requested symbols for all included $\mathtt{com}$(s).
    
    \item \textbf{Confirmation:} An availability block $B_a$ is confirmed once it is followed by $k'$ subsequent availability blocks. 

    \item \textbf{Response:} For each call to \textsc{Verify}($\mathtt{com}$), output \textsc{Verify}($\mathtt{com}$) $= 1$ if $\mathtt{com}$ is a confirmed $\mathtt{com}$ in the availability chain; otherwise, output $0$.
\end{itemize}



\subsection{Diffusion Mining}\label{diffusion_maintext}

Our initial idea for extending the above design into a sharding protocol is to have all nodes collectively mine a single proposer chain, while nodes within each shard mine their respective availability chains. Specifically, the proposer chain references transaction blocks across all shards, whereas the $S$ availability chains operate in parallel, with each node producing availability blocks exclusively for its assigned shard. However, each shard must have sufficient nodes to ensure an honest majority. To overcome this limitation, we introduce a diffusion mining that secures each shard even when only one honest node is present.

Diffusion mining classifies availability blocks into two types: local availability blocks and global availability blocks. While a local availability block operates within a single shard, a global availability block can extend all chains across shards. In a longest chain-based protocol, the number of valid blocks reflects hashing power, and security is maintained as long as honest blocks outnumber adversarial ones. Diffusion mining allows honest nodes to distribute their hashing power across shards, thereby securing those with weaker honest hashing power. 
We adopt this mechanism to strengthen our design:

\noindent {\bf Availability chains (\textsc{Verify}).} Each $\mathtt{com}$ inherits a shard ID \( i \) from its corresponding transaction block. Nodes within shard \( i \) mine an availability chain, where each availability block references only the $\mathtt{com}$(s) associated with shard \( i \). 

\begin{itemize}
    \item \textbf{Sampling:} For any confirmed $\mathtt{com}$ tagged $i'$ not in the longest availability chain of shard $i$,
    \begin{enumerate}
        \item If $i'=i$, given its level $l$ and the current highest level $l'$ of the proposer chain, a node randomly requests $\frac{l'-l}{S}$ symbols of the block.
        \item If \( i' \neq i \), a node randomly requests a constant number \( u \) of symbols from the block.
    \end{enumerate}
    \item \textbf{Packing:} Each node gathers all unreferences $\mathtt{com}$(s) in shard $i$ for which the requested symbols have been received. These are used to populate the $\mathtt{avai\_com\_set}$ of an availability block. The node collects the highest availability block hashes across all shards and records it in $\mathtt{global\_parents}$.
    \item \textbf{Mining:} Each node searches for a valid $\mathtt{nonce}$ to solve the PoW puzzle and produces an availability block with hash $h$ if they find a valid $\mathtt{nonce}$. The type of the availability block is determined by comparing its size with a threshold \( h' \).
    \begin{enumerate}
        \item If $h > h'$, the block is a local availability block and extends the availability block with hash $\mathtt{inter\_parent}$ in shard~$i$.
        \item If $h \leq h'$, the block is a global availability block and extends the availability blocks whose hashes are included in $\mathtt{global\_parents}$ across shards.
    \end{enumerate}

    \item \textbf{Verification:} A node accepts an availability block if \begin{enumerate}
        \item The PoW solution is valid.
        \item The node has received the requested symbols for all referenced $\mathtt{com}$(s).
    \end{enumerate}
    \item \textbf{Confirmation:} An availability block $B_a$ is confirmed once it is followed by $k'$ subsequent availability blocks. 
    \item \textbf{Response:} For each call to \textsc{Verify}($\mathtt{com}_i$), output \textsc{Verify}($\mathtt{com}_i$) $= 1$ if $\mathtt{com}_i$ is confirmed in the $i$-th availability chain; otherwise, output $0$.
\end{itemize}

\begin{figure}[!ht]
 \footnotesize
 \begin{boxedminipage}[t]{\columnwidth}
 \underline{\textbf{Parameters}}: $H$ and $G$ denote hash functions, $\delta$ is the overall mining difficulty, and $\delta_p$, $\delta_{ga}$, $\delta_{la}$, and $\delta_o$ specify the difficulty thresholds for different block types.

  \underline{\textbf{Packing}}: The node gathers the necessary information for a hybrid block:

\begin{itemize}
    \item Block header $B_h$: \begin{enumerate}
        \item $\mathtt{shard\_ID}$: The node's affiliated shard.
        \item $\mathtt{prop\_parent}$: The hash of the highest proposer block.
        \item $\mathtt{inter\_parent}$: The hash of the highest availability block in the $\mathtt{shard\_ID}$-th shard.
        \item $\mathtt{global\_parents}$: A set containing the hashes of the highest availability blocks in their respective shards.
        \item $\mathtt{prop\_root}$: The root of a Merkle tree generated from $\mathtt{prop\_com\_set}$.
        \item $\mathtt{avai\_root}$: The root of a Merkle tree generated from $\mathtt{avai\_com\_set}$.
        \item $\mathtt{order\_root}$: The root of a Merkle tree generated from $\mathtt{avai\_blk\_set}$.
        \item $\mathtt{com}$: The root of a CMT generated from $\mathtt{data\_blob}$.
    \end{enumerate}
    \item Block body $B_b$:
\begin{enumerate}
    \item $\mathtt{prop\_com\_set}$: Set of transaction block commitments (\(\mathtt{com}\)) not yet referenced in the longest proposer chain.
    \item $\mathtt{avai\_com\_set}$: Set of \(\mathtt{com}\)s that (i) are unreferenced in the longest availability chain of \(\mathtt{shard\_id}\)-th shard, (ii) are confirmed in the proposer chain, and (iii) have all requested symbols received.
    \item $\mathtt{avai\_blk\_set}$: Set of hash of confirmed availability blocks not yet referenced in the longest ordering chain.
    \item $\mathtt{data\_blob}$: A set of \(t\) valid transactions.
\end{enumerate}
\end{itemize}

\underline{\textbf{Mining}}: 

\begin{enumerate}
    \item The node iterates $\mathtt{nonce}$ to make:
    \begin{equation}
    \begin{split}\label{PoW}
        \mathtt{hash}&=H(G(B_h), \mathtt{nonce}) \leq \delta.
    \end{split}
    \end{equation}
    \item When a valid $\mathtt{nonce}$ is found, determine the block type as follows, and broadcast $B_h$ along with the corresponding field in $B_b$ to other nodes:  
    \begin{itemize}
        \item $\mathtt{hash} \leq \delta_{ga}$: global availability block.
        \item $\delta_{ga} < \mathtt{hash} \leq \delta_{la}$: local availability block.
        \item $\delta_{la} < \mathtt{hash} \leq \delta_{o}$: ordering block.
        \item $\delta_{o} < \mathtt{hash} \leq \delta_{p}$: proposer block.
        \item $\delta_{p} < \mathtt{hash} \leq \delta$: transaction block.
    \end{itemize}
\end{enumerate}

\underline{\textbf{Sampling}}: For each $\mathtt{com}$ in the proposer chain, whose associated proposer block is followed by \( k \) blocks and originates from shard \( i \), a node in shard \( i' \) requests:
\begin{itemize}
    \item \( u \) symbols if \( i' \neq i \);
    \item \( \frac{l' - l}{S} \) symbols if \( i' = i \), where \( l \) is the level of the proposer block and \( l' \) is the highest level in the proposer chain.
\end{itemize}

\underline{\textbf{Verification}}: A node validates a block by verifying its PoW solution and inspecting its content according to the block type:

\begin{itemize}
  \item \textbf{Availability block}: accepted if
  \begin{enumerate}
    \item the Merkle root of $\mathtt{avai\_com\_set}$ matches $\mathtt{avai\_root}$ in \(B_h\);
    \item all $\mathtt{com}$s are confirmed in the proposer chain;
    \item all requested symbols are received.
  \end{enumerate}

  \item \textbf{Ordering block}: accepted if 
  \begin{enumerate}
    \item the Merkle root of $\mathtt{avai\_blk\_set}$ matches $\mathtt{order\_root}$ in \(B_h\);
    \item all availability blocks in $\mathtt{avai\_blk\_set}$ are confirmed.
  \end{enumerate}

  \item \textbf{Proposer block}: accepted if the Merkle root of $\mathtt{prop\_com\_set}$ matches $\mathtt{prop\_root}$ in \(B_h\).

  \item \textbf{Transaction block}: accepted without additional verification.
\end{itemize}

\underline{\textbf{Confirmation}}. Proposer, availability, and ordering blocks are confirmed after depths of $k$, $k'$, and $k''$, respectively.

 \end{boxedminipage}
 \caption{The full protocol of OptChain.}
 \label{fig:OptChain_algorithm}
\end{figure}

Both in-shard and out-shard nodes request symbols for a specific confirmed $\mathtt{com}$ in the proposer chain. The key difference is that in-shard nodes continuously sample an increasing number of symbols over time until $\mathtt{com}$ is included in the longest availability chain, whereas out-shard nodes request a fixed number of symbols at once. Moreover, since out-shard nodes greatly outnumber in-shard nodes, each out-shard node only needs to request a small number of symbols, making the communication cost for verifying out-shard blocks negligible.


\subsection{Cross-shard Atomicity}\label{cross_shard_tx_atomicity_maintext}

We adopt the Unspent Transaction Output (UTXO) model, where a transaction is defined as $\mathtt{tx}:\{I_1,\ldots,I_k; O_1,\ldots,O_h\}$. The $\{I_i\}$ represents withdrawal operations, while the output set $\{O_j\}$ corresponds to deposit operations, and each input/output is associated with an address. Outputs can serve as inputs to subsequent transactions but can be spent only once. In the context of sharding, each address is assigned a shard ID, inherited by its inputs and outputs. A transaction whose inputs and outputs belong to the same shard is termed a \textit{domestic transaction (domestic-tx)}, whereas one spanning multiple shards is a cross-tx. Shards whose IDs match any input's shard ID are called \textit{withdrawal shards} for the cross-tx, and those matching any output's shard ID are \textit{deposit shards}. Each cross-tx is replicated across all withdrawal and deposit shards and processed independently as \textit{withdrawal-txs} and \textit{deposit-txs}, respectively.






In sharding protocols, nodes must reach consensus on whether to commit a cross-tx—a requirement formalized as the cross-shard atomicity property, defined in Definition~\ref{cross_shard_atomicity_definition}. A common approach to ensure this is the 2PC protocol~\cite{2pc}. Specifically, coins are locked once a withdrawal-tx is confirmed. After all withdrawal shards successfully lock the coins, the deposit shards accept the corresponding deposit-txs and create new coins. An important question in sharding protocols is how nodes can verify coin locking without access to other shards' ledgers. Protocols such as Monoxide~\cite{Monoxide} and Manifoldchain~\cite{Manifoldchain} address this by Merkle proofs of withdrawal-txs generated by honest nodes in those shards. Once a withdrawal-tx is confirmed in the longest chain, all honest nodes recognize it and mark the corresponding coins as locked.

However, this mechanism cannot be directly applied to OptChain. In OptChain, nodes do not verify transactions when confirming blocks; instead, they do it after retrieving the complete transaction blocks, removing invalid and duplicate transactions to sanitize the ledger. This design introduces two new challenges:

\begin{itemize}
    \item A withdrawal-tx in the longest chain is not necessarily valid, as it may be a duplicate or double-spend transaction. An adversary can craft such invalid withdrawal-txs and exploit their Merkle proofs to generate valid deposit-txs.
    \item Network delays may cause some honest nodes to confirm deposit-txs before their corresponding withdrawal-txs, leading them to reject the deposits, while others accept them, resulting in a split view.
\end{itemize}


\begin{figure*}
  \centering
  \includegraphics[width=14cm]{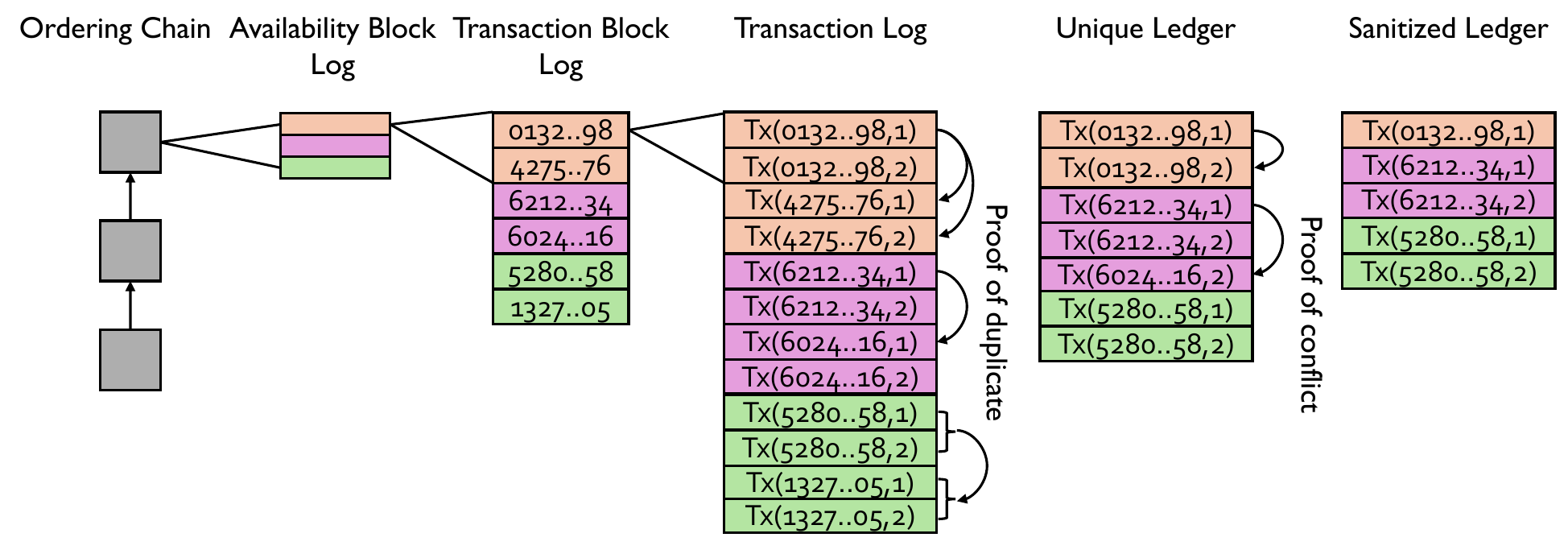}
  \vspace{-2mm}
  \caption{The process of ledger sanitization. First, availability block hashes are arranged through the ordering chain to form an availability block log. Second, each availability block’s \(\mathtt{com}\)(s) are appended in sequence to create a transaction block log. Third, each \(\mathtt{com}\) is expanded by block size and paired with transaction indices as \((\mathtt{com}, i)\), forming a transaction log. Upon retrieving a transaction block, its transactions follow the same order as in the CMT, ensuring that any two nodes retrieving \((\mathtt{com}, i)\) obtain the same transaction. Finally, entries with \(\mathtt{proof\_of\_duplicate}\) or \(\mathtt{proof\_of\_conflict}\) are removed, yielding a sanitized log containing only unique and valid transactions.}
  \label{fig:ledger}
  \vspace{-6mm}\hspace{-10mm}
\end{figure*}

To address these challenges, we introduce an ordering chain that sequences availability blocks across shards, thereby establishing a global transaction order. This mechanism guarantees that any valid deposit-tx is confirmed only after all corresponding withdrawal-txs. Furthermore, we leverage fraud proofs~\cite{fraudproof} to sanitize the ledger by removing duplicate and invalid transactions. 

\subsubsection{Ordering Chain}

Mining an ordering block is straightforward: package the hashes of confirmed availability blocks into an ordering block, then perform PoW to append it to the longest ordering chain:
\begin{itemize}
    \item \textbf{Packing}: Nodes collect hashes of confirmed availability blocks not already included in ancestor ordering blocks to form $\mathtt{avai\_blk\_set}$, and select the hash of the latest ordering block as $\mathtt{order\_parent}$.
    \item \textbf{Mining}: Nodes iterate over $\mathtt{nonce}$ values to solve the PoW puzzle and mine a valid ordering block, which is broadcast to the network.
    \item \textbf{Verification}: An ordering block is accepted if its PoW is valid, and all referenced availability blocks are confirmed and not included in any ancestor blocks.
    \item \textbf{Confirmation}: An ordering block is confirmed once it is followed by $k''$ ordering blocks.
\end{itemize}

The structure of all the chains is illustrated in Fig.~\ref{fig:full_structure}. 
To prevent the adversary from concentrating hashing power on a single chain to exceed its \( 1/2 \) fault tolerance threshold, we employ the idea of 2-for-1 PoW~\cite{Bitcoinbackbone} to mine these \( S+2 \) chains concurrently. Nodes first attempt to mine a \textit{hybrid block} containing all necessary information to serve as any block type. Once mined, the block is classified into a specific type according to the region where its hash falls. We present the complete protocol of OptChain in Fig.~\ref{fig:OptChain_algorithm}.

\subsubsection{Fraud Proof}

When an adversary attempts to use an invalid withdrawal-tx to construct a fraudulent Merkle proof for a deposit-tx, honest nodes can generate a fraud proof to alert other shards of the invalidity. Specifically, there are two types of fraud proofs. The first type, $\mathtt{proof\_of\_conflict}$, is designed to filter out invalid transactions, which are typically produced by corrupted nodes, and is formatted as follows:

\begin{equation}
\begin{split}\label{pow_eq_1}
\mathtt{proof\_of\_conflict} = &\\
\{\mathcolorbox{green!15}{(\{\mathtt{tx}_1^0, \mathtt{tx}_2^0, \dots\},}& \mathcolorbox{green!15}{\mathtt{com}_0,\mathtt{merkle\_proof})}, \\
\mathcolorbox{red!15}{(\{\mathtt{tx}_1^1, \mathtt{tx}_2^1, \dots\}, }&\mathcolorbox{red!15}{\mathtt{com}_1, \mathtt{merkle\_proof}_1),}\\
\mathcolorbox{red!15}{(\{\mathtt{tx}_1^2, \mathtt{tx}_2^2, \dots\}, }&\mathcolorbox{red!15}{\mathtt{com}_2, \mathtt{merkle\_proof}_2),}\\
\dots \\
\mathcolorbox{red!15}{(\{\mathtt{tx}^i_1, \mathtt{tx}^i_2, \dots\}}&\mathcolorbox{red!15}{, \mathtt{com}_i,\mathtt{merkle\_proof}_i)}\},
\end{split}
\end{equation}
where $\mathtt{tx}^i_k$ denotes the $k$-th transaction in the block rooted at~$\mathtt{com}_i$. A $\mathtt{proof\_of\_conflict}$ comprises two components. The first component (green) contains valid transactions ($\mathtt{tx}_1^0, \mathtt{tx}_2^0, \dots$) against which no valid fraud proofs exist. The inclusion of transactions sharing the superscript~$i$ is authenticated via a single Merkle proof, $\mathtt{merkle\_proof}_i$. The second component (red) contains the conflicting transactions, organized similarly. Nodes validate the fraud proof by verifying whether any conflicting transaction double-spends inputs from the first component or contains invalid signatures. 

For duplicate transactions (which are also invalid), nodes do not need to know the full transaction data. So we replace each transaction with its hash  and introduce a lightweight $\mathtt{proof\_of\_duplicate}$, formatted as follows:

\begin{equation}
\begin{split}\label{pow_eq_2}
\mathtt{proof\_of\_duplicate} = &\\
\{\mathcolorbox{green!15}{(\{\mathtt{hash}_1^0, \mathtt{hash}_2^0, ...\}, }&\mathcolorbox{green!15}{\mathtt{com}_0, \mathtt{merkle\_proof}_0),}\\
\mathcolorbox{red!15}{(\{\mathtt{hash}_1^1, \mathtt{hash}_2^1, ...\}, }&\mathcolorbox{red!15}{\mathtt{com}_1, \mathtt{merkle\_proof}_1),}\\
\dots\\
\mathcolorbox{red!15}{(\{\mathtt{hash}_1^i, \mathtt{hash}_2^i, ...\}, }&\mathcolorbox{red!15}{\mathtt{com}_i, \mathtt{merkle\_proof}_i)}\}
\end{split}
\end{equation}

$\mathtt{hash}^i_k$ denotes the hash of the $k$-th transaction in the block rooted at~$\mathtt{com}_i$. The first component (green) contains the hashes of valid transactions ($\mathtt{hash}_1^0, \mathtt{hash}_2^0, \dots$) against which no valid fraud proofs exist. The second component (red) contains the hashes of duplicate transactions. Nodes validate the fraud proof by checking for any intersection of hashes between these two components.

Both $\mathtt{proof\_of\_conflict}$ and $\mathtt{proof\_of\_duplicate}$ fall under the category of fraud proofs. 
Upon receiving a fraud proof, out-shard nodes verify its correctness and mark the implicated transaction as invalid. The communication overhead of transmitting fraud proofs is minimal. Invalid transactions are typically grouped in an invalid block created by corrupted nodes, which can be proven invalid using a single Merkle proof. Fraud proofs for duplicate transactions add only negligible overhead since they contain only lightweight hash digests.

\subsubsection{Ledger Sanitization}\label{ledger_sanitization}

We construct a sanitized ledger log by sequentially transforming block information into a consistent and ordered transaction log, as shown in Fig.~\ref{fig:ledger}. 

\begin{enumerate}
    \item \textbf{Ordering availability blocks via the ordering chain.} Each ordering block contains multiple hashes of availability blocks. These hashes are added to a queue in the order they appear within each ordering block and in the sequence that the ordering blocks appear in the ordering chain. The resulting queue forms an availability block log, where each entry corresponds to an availability block hash.

    \item \textbf{Transferring availability blocks to transaction blocks.} Each availability block contains multiple \(\mathtt{com}\)s of transaction blocks. Append these \(\mathtt{com}\)s to a queue in the order they appear in their respective availability blocks, and in the order those availability blocks appear in the availability block log. This queue forms a transaction block log with multiple \(\mathtt{com}\)s as its entries.

    \item \textbf{Transferring transaction blocks to transactions.} Each \(\mathtt{com}\) is the root of a CMT generated from \(t\) transactions. We use \((\mathtt{com}, i)\) to denote the \(i\)-th transaction in the block rooted at \(\mathtt{com}\)\@. These pairs are appended to a queue in the order that \(\mathtt{com}\) appears in the transaction block log, and in increasing order of \(i\)\@. This queue forms the finalized ledger log. When a node retrieves a transaction block, its transactions are ordered as they were in the CMT. Therefore, any two nodes retrieving \((\mathtt{com}, i)\) will get the same transaction.

    \item \textbf{Removing duplicate and invalid transactions.} First, remove entries with valid \texttt{proof\_of\_duplicate} from the transaction log. Then remove entries with \texttt{proof\_of\_conflict} from the remaining log. The result is a sanitized ledger containing only unique and valid transactions.

\end{enumerate}

\subsubsection{Cross-shard Transaction Execution}

The complete protocol for processing cross-txs is as follows:

\begin{enumerate}
    \item \textbf{Initialization}. The client generates a cross-tx and decomposes it into multiple withdrawal-txs and deposit-txs. Subsequently, it broadcasts the withdrawal-txs to the corresponding withdrawal-shards.
    
    \item \textbf{Withdrawal}. Upon retrieving a withdrawal-tx from the confirmed log, a node in a withdrawal shard checks its validity using local ledger data. If valid, the node generates a Merkle proof and sends it to the user; otherwise, it is discarded.

    \item \textbf{Deposit}. After collecting all Merkle proofs, the user sends them with the deposit-tx to the deposit shards. When a node in a deposit shard receives or retrieves a deposit-tx, it checks: (a) all Merkle proofs are present, and (b) no valid fraud proof exists. If valid, the node accepts and executes the deposit-tx; otherwise, it is rejected and discarded.
\end{enumerate}

We prove in Appendix~\ref{cross_tx_atomicity_proof} that the above protocol guarantees cross-shard atomicity defined in Definition~\ref{cross_shard_atomicity_definition}.

\section{Analysis}

Let $\alpha_i$ and $\beta_i$ denote the honest and adversarial fractions in shard $i$. Define $\lambda_p$, $\lambda_t$, $\lambda_i$, $\lambda_s$, and $\lambda_o$ as the mining rates (blocks per second) of proposer, transaction, local availability, global availability, and ordering blocks, respectively. Let $\Delta_p$, $\Delta_t$, $\Delta_i$, $\Delta_s$, and $\Delta_o$ represent their corresponding maximum network delays. Since proposer, availability, and ordering blocks contain only lightweight digests, $\Delta_p = \Delta_s = \Delta_i = \Delta_o = \underline{\Delta} \leq O(\tfrac{b}{\alpha n})$. In contrast, transaction blocks carry full ledger data, yielding $\Delta_t = \overline{\Delta} = O(b)$.

\subsection{Security of OptChain}\label{safety_liveness}

\textit{OptChain's security inherently reduces to that of Bitcoin.} Intuitively, our protocol optimizes block availability verification while strictly preserving Bitcoin's underlying consensus layer. Consider a naive reduction: if nodes were required to download full transaction blocks to vote on their availability, each availability chain would directly mirror Bitcoin, and the aggregate of all availability chains would be equivalent to Manifoldchain. The proposer chain implements PVID, enabling nodes to verify block availability using only block fragments. Crucially, this does not alter Bitcoin's core assumption: when an honest node receives a full block (verifying its availability), it broadcasts that block so every other honest node can also verify it within $\Delta$. This assumption is also guaranteed by our PVID's agreement property, as defined in Section~\ref{prelim_VID}. Furthermore, the ordering chain solely establishes a global sequence across availability blocks to ensure cross-shard atomicity, leaving the underlying consensus mechanics entirely unmodified. Formally, OptChain satisfies both safety and liveness, as established in the following theorems:

\begin{theorem}[Safety]
\label{OptChain_security_main_text}
If there is at least one honest node in each shard, OptChain satisfies safety except with negligible error probability $\emph{negl}(k)$ provided that:
\begin{equation}
\begin{split}
p_p =& \alpha e^{-\lambda_p\underline{\Delta}} > \frac{1}{2},p_o = \alpha e^{-\lambda_o\underline{\Delta}} > \frac{1}{2},\\
p_a =& \frac{\lambda_s \underline{\alpha} + S\lambda_i\alpha'}{\lambda_s+S\lambda_i} e^{-(\lambda_s+\lambda_i)\underline{\Delta}} > \frac{1}{2},\\
\alpha' =& \alpha\left[1-\left(\frac{1}{v}\right)^{k}\frac{1}{v-1}\right] > \frac{1}{2},    \\
\end{split}
\end{equation}
where $\underline{\alpha} = \alpha'-\log^2(\alpha')$.
\end{theorem}
\begin{theorem}[Liveness]
\label{OptChain_security_main_text_2}
Suppose each shard contains at least one honest node and the conditions of Theorem~\ref{OptChain_security_main_text} hold. Assuming a current ledger height $\ell_{cur}$, a transaction at height $\ell_{cur} - \ell$ ($\ell_{cur} > \ell$) submitted at time $t$ is confirmed by $t+u$ except with a negligible error probability \emph{negl}($\ell$), where
\[
u = \sum_{x \in \{p, o, a\}} \frac{k + \lceil1/q_x\rceil}{g_x} + 5\underline{\Delta},
\]
where $g_x$ and $q_x$ represent the chain growth and chain quality for each chain $x \in \{p, o, a\}$, corresponding to the proposer, ordering, and availability chains, respectively:
\begin{align*}
    g_x &= (1-\delta)\frac{p_x \Lambda_x}{1+p_x\underline{\Delta}\Lambda_x}, \\
    q_x &= 1-(1+\delta)\frac{1+p_x(1-\alpha_x)\underline{\Delta}\Lambda_x}{p_x}.
\end{align*}
Here, $\Lambda_p = \lambda_p$, $\Lambda_o = \lambda_o$, and $\Lambda_a = \lambda_s + \lambda_i$; while $\alpha_p = \alpha_o = \alpha$, and $\alpha_a = \underline{\alpha}$.
\end{theorem}
While individual chains may require varying confirmation depths (e.g., $k$, $k'$, and $k''$), we simplify our analysis using a unified depth $k = \max\{k, k', k''\}$. In Theorem~\ref{OptChain_security_main_text}, $p_p$, $p_o$, and $p_a$ represent the effective honest ratios for mining these chains while accounting for network delay. It suggests that safety holds as long as these effective honest ratios maintain a majority. This constraint is consistent with those established in prior works~\cite{Manifoldchain,Bitcoinmadsimple}. To satisfy this condition in the presence of a larger network delay, mining rates must be correspondingly reduced. Notably, baseline protocols such as Bitcoin and Manifoldchain incur a substantial network delay $\overline{\Delta}$ that scales linearly with transaction block size. Conversely, because OptChain's block size is restricted to its metadata, it experiences only a negligible delay $\underline{\Delta}$, thereby permitting significantly higher mining rates. This explains why OptChain outperforms other protocols in throughput. Theorem~\ref{OptChain_security_main_text_2} states that any transaction submitted at time $t$ will be confirmed by $t+u$. Here, $u$ is calculated based on the chain growth and chain quality parameters of these chains, representing the worst-case latency for any transactions.

\noindent \textbf{Proof Sketch.} We first show that the proposer chain, ordering chain, and all availability chains satisfy the three standard properties of PoW-based protocols defined in Appendix~\ref{pow_properties}: \textit{Common Prefix (CP)}, \textit{Chain Quality (CQ)}, and \textit{Chain Growth (CG)}. Based on these properties, honest nodes agree on a consistent and totally ordered log of block commitments. We then prove that OptChain also satisfies all properties of the VID scheme defined in Definition~\ref{VID_security_properties}, which ensures that all blocks committed in the log are available. Consequently, the agreed log effectively becomes a log of executable transactions once honest nodes retrieve the full blocks. Therefore, we conclude that our protocol satisfies both safety and liveness. The detailed proof is deferred to Appendix~\ref{security_proof}.


\subsection{OptChain's Throughput}\label{OptChain_tps}

We present the upper bound of OptChain's throughput as follows:

\begin{theorem}\label{opt_tps}
Given a maximum error probability of $\sigma$, the throughput of \textit{OptChain} is bounded by $T = \frac{\overline{C}}{1 - {(\sigma - \emph{negl}(k))}^{\frac{1}{\alpha n}}}$.
\end{theorem}

As $k \rightarrow \infty$, $T \rightarrow T^*$.

Given a target error probability $\sigma$, we demonstrate how to configure OptChain to closely approach the optimal throughput $T^*$ established in Theorem~\ref{opt_tps}. We provide a proof sketch demonstrating how this configuration achieves the claimed throughput and defer the detailed proof to Appendix~\ref{proof_opt_tps}.

\noindent {\bf Step 1: mining rate for transaction blocks ($\lambda_t$).} To achieve a throughput of $T$, we set $\lambda_t \cdot t = T$. While the specific values of $\lambda_t$ and $t$ are flexible, preferring a larger $t$ and smaller $\lambda_t$ helps reduce the number of $\mathtt{com}$ entries per proposer block, resulting in smaller proposer block size.

\noindent {\bf Step 2: number of shards ($S$).} Since each shard contributes at most $\overline{C}$ to the total throughput, we set $S = \frac{T}{\overline{C}}$ to ensure the availability property—guaranteeing that every incoming transaction per second can be downloaded by at least one honest node.

\noindent {\bf Step 3: mining rate for proposer blocks ($\lambda_p$).} To embed the commitment of transaction blocks into the proposer chain, we set $\lambda_p \cdot e = \lambda_t$, where $e$ denotes the maximum number of transaction blocks referenced per proposer block. While the specific values of $\lambda_p$ and $e$ are flexible, a larger $\lambda_p$ and smaller $e$ are preferred to minimize the proportion of deceived honest nodes $\sigma$, as defined in Theorem~\ref{theorem_availability_consistency}. Accordingly, we choose $\lambda_p$ as large as possible, subject to it being less than both $\lambda_t$ and the upper bound specified in Theorem~\ref{OptChain_security_main_text}. Once $\lambda_p$ is set, $e$ is chosen to ensure $\lambda_p \cdot e = \lambda_t$.

\noindent {\bf Step 4: mining rates for availability blocks ($\lambda_s$ and $\lambda_i$).} To keep pace with the proposer chain, we set $\lambda_s + \lambda_i = \frac{\lambda_p}{S}$. We then choose $\lambda_s$ and $\lambda_i$ to satisfy the constraint in Theorem~\ref{OptChain_security_main_text}, which requires $\frac{\lambda_s\underline{\alpha} + S\lambda_i\alpha}{\lambda_s + S\lambda_i} > \frac{1}{2}\textsc{exp}\left((\lambda_s + \lambda_i)\underline{\Delta}\right)$. As $\lambda_s + \lambda_i$ is fixed, increasing $\lambda_i$ raises the left-hand side but also increases the number of global availability blocks, resulting in higher communication overhead. Therefore, we choose $\lambda_i$ as small as possible while still satisfying the constraint, and set $\lambda_s = \frac{\lambda_p}{S} - \lambda_i$.

\noindent {\bf Step 5: Mining rate for ordering blocks ($\lambda_o$).} Suppose each ordering block can reference up to $r$ availability blocks. Since the ordering chain globally order availability blocks across shards, we set $\lambda_o r = (\lambda_s + \lambda_i)S = \lambda_p$. The values of $\lambda_o$ and $r$ are flexible; a larger $r$ and smaller $\lambda_o$ are preferred to reduce the forking rate.

\noindent {\bf Step 6: confirmation depths $k$.} The parameter $k$ is selected to satisfy the constraint $\alpha' > \frac{1}{2}$ in Theorem~\ref{OptChain_security_main_text} and to minimize the error probability $\texttt{negl}(k)$, determined by whichever condition is more stringent.
 
\noindent {\bf Step 7: parameters for erasure code ($v$ and $d$).} As within each shard there are $\lambda_s + \lambda_i$ availability blocks being confirmed per second, this indicates that every honest node has received the requested symbols of the corresponding transaction blocks, implying that $(\lambda_s + \lambda_i) \cdot e \cdot b \cdot \frac{k}{d} < \underline{C}$. Furthermore, since $\lambda_s + \lambda_i = \frac{\lambda_p}{S} = \frac{\overline{C}}{T} \cdot \frac{\lambda_t}{e} = \frac{\overline{C}}{e b}$, we set $d \geq \frac{\overline{C}}{\underline{C}} \cdot k$. Intuitively, the proposer chain's mining rate reflects the speed at which the fastest node downloads the full transaction block of $d$ symbols, while the availability chains' rate reflects the speed at which the slowest node downloads any $k$ of the $d$ symbols. $v$ is chosen as large as possible to minimize the error probability in Theorem~\ref{OptChain_security_main_text}.

\noindent \textbf{Proof Sketch.} Under the parameter configuration described above, we bound the probability that at least one security property is violated by considering three distinct failure events:

\begin{enumerate}
    \item \textbf{Adversarial majority.} The fraction of corrupted nodes (including deceived honest nodes) exceeds $50\%$. This occurs with probability $\sigma_1 = \texttt{negl}(k)$, as shown in Appendix~\ref{app_availability_security}.
    \item \textbf{Honest absence.} Given $n - f$ honest nodes and $T/\overline{C}$ shards, the probability that a specific transaction's shard contains no honest nodes is $\sigma_2 = \left(1 - \frac{\overline{C}}{T} \right)^{\alpha n}$, as derived in Appendix~\ref{proof_opt_tps}.
    \item \textbf{Chain property violation.} One or more of the CP, CQ, or CG properties are violated. This occurs with probability $\sigma_3 = \texttt{negl}(k)$, as shown in Appendix~\ref{proof_step_1}.
\end{enumerate}

Therefore, the total error probability $\sigma$ is bounded by:
\[
\sigma \leq \sigma_1 + \sigma_2 + \sigma_3 = \left(1 - \frac{\overline{C}}{T}\right)^{\alpha n} + \texttt{negl}(k).
\]
Solving for $T$, we derive the OptChain throughput $\frac{\overline{C}}{1 - {(\sigma-\texttt{negl}(k))}^{\frac{1}{\alpha n}}}$, which matches the bound stated in Theorem~\ref{opt_tps}.

\section{Experiments}

\noindent {\bf Testbed.} We implemented a prototype of the OptChain client in $7,000$ lines of Rust code~
\cite{optchain_anon} and deployed it on Amazon Elastic Compute Cloud (EC2) to evaluate its performance in a realistic geo-distributed environment. 
We build a consensus network of $64$ nodes, each node operates on an EC2 \texttt{t3.medium} instance equipped with 2 vCPU cores, 4GB of RAM, and a 20GB NVMe SSD. To ensure geographic diversity, these instances are distributed across four major regions: London, Virginia, Sao Paulo, and Tokyo. The transactions and all metadata (e.g., hashes and signatures) are implemented with the same data structures to ensure a fair comparison.

\noindent \textbf{Baselines.} We compare OptChain against Manifoldchain~\cite{Manifoldchain} and Prism~\cite{Prism}. Manifoldchain represents the SOTA in sharding protocols that offer a security level comparable to ours (i.e., $1/2$ global fault tolerance). We exclude other permissionless sharding protocols, such as Elastico\cite{Elastico}, Omniledger\cite{Omniledger}, and RapidChain\cite{Rapidchain}, as their lower fault tolerance (typically $< 1/3$ or $1/4$) precludes a fair direct comparison. Conversely, Prism represents the SOTA in non-sharding permissionless protocols. By benchmarking against both, we demonstrate that OptChain outperforms the leading protocols in both the sharding and non-sharding categories. To ensure a fair comparison (i.e., using the same P2P networking and transaction size), we emulate Prism by tuning OptChain's parameters. Specifically, by setting the shard count to one and configuring nodes to request full transaction blocks rather than coded symbols for each $\mathtt{com}$, OptChain functionally reduces to Prism (specifically, a variant without voter chains) and achieves the same throughput as the vanilla Prism protocol. For Manifoldchain, we deploy its open-sourced implementation\cite{che2025manifoldchainrepo} on our testbed to ensure identical experimental configurations. We also compare OptChain with the theoretical optimal throughput. To account for the inherent bandwidth volatility of AWS EC2, we estimate the theoretical optimal throughput via an empirical micro-benchmark rather than static analysis. We use the \texttt{tc} command to set downlink bandwidths based on a real-world dataset~\cite{realbandwidth}, selectively sampling values between 10 Mbps and 80 Mbps. We then run iterative tests where a random sender broadcasts transactions to all receivers for one minute. We compute the average download transactions for each node across multiple iterations. The theoretical optimal throughput is calculated as the number of download transactions of the $x$ highest-performing nodes per second, where $x$ corresponds to the maximum shard count.

\noindent \textbf{Evaluation.} In our first experiment, we evaluate OptChain's maximum throughput under a fixed error probability and compare it against the baselines. Subsequently, to validate OptChain's design for optimal horizontal and vertical scalability, we evaluate these two dimensions independently in the second and third experiments. 

\subsection{Throughput under Error Probabilities}

This experiment evaluates the throughput of OptChain across a range of maximum allowable error probabilities (error threshold), spanning from $e^{-50}$ to $e^{-2}$. We deploy the system on $64$ EC2 instances, each hosting a single node. Having demonstrated OptChain's superior throughput in our 64-node evaluation, we defer its deployment in a large-scale permissionless network to future work. Consistent with our theoretical benchmark configuration, we emulate realistic network heterogeneity by limiting each node's downlink bandwidth using the \texttt{tc} command. These bandwidth limits are assigned by uniformly sampling values between $10$ Mbps and $80$ Mbps from a real-world dataset~\cite{realbandwidth}. For each error probability threshold, we optimize the parameters of both OptChain and the baselines to maximize their respective throughput. 

\begin{figure}
  \centering
  \includegraphics[width=8cm]{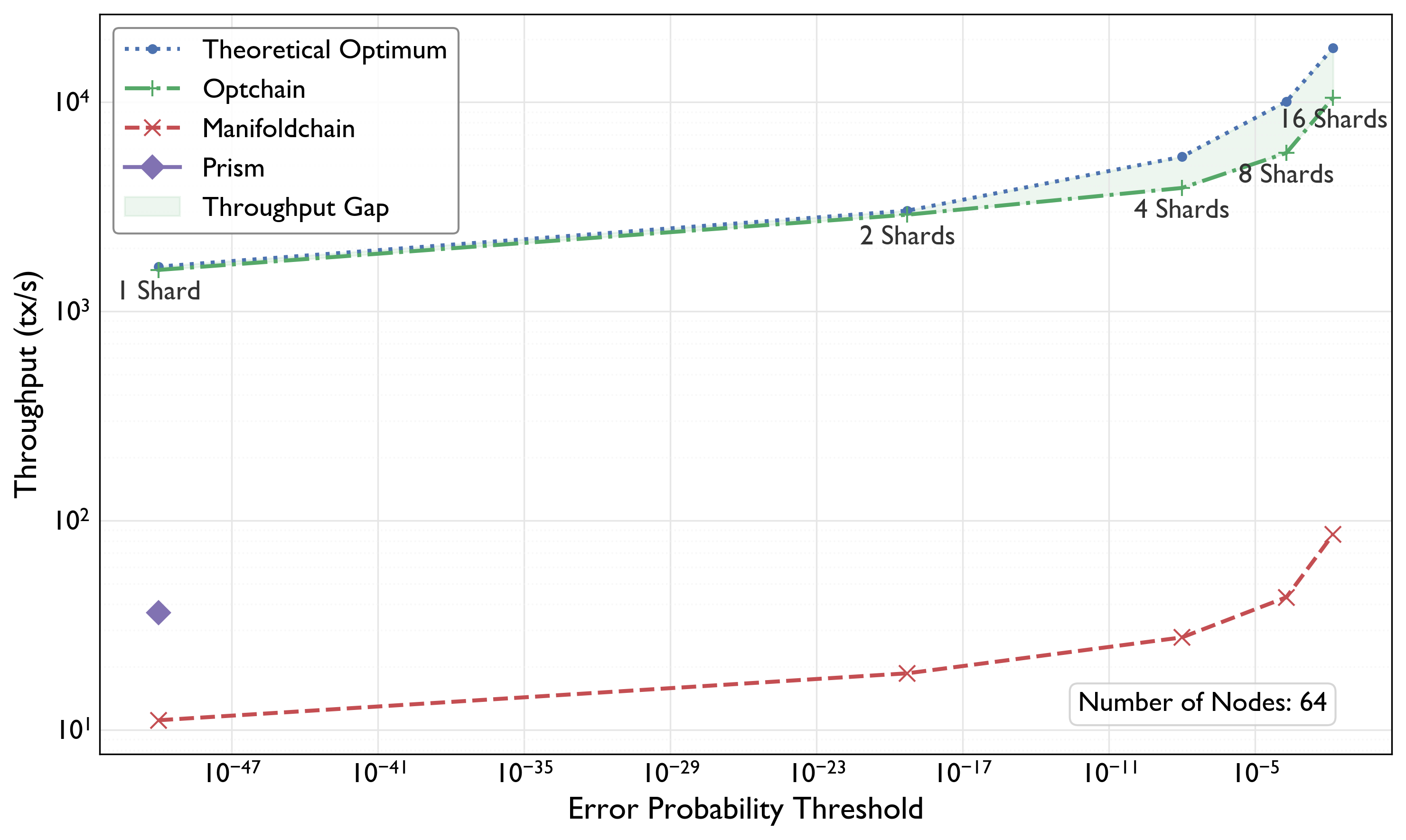}
  \vspace{-2mm}
  \caption{Throughput comparison of OptChain against the theoretical optimum and state-of-the-art baselines across varying error probability thresholds. OptChain closely tracks the theoretical optimum, whereas Manifoldchain demonstrates substantially lower throughput. Prism appears as a single point because its non-sharding throughput is independent of the error threshold.}
  \label{fig:exper_1}
  \vspace{-6mm}\hspace{-10mm}
\end{figure}

Figure~\ref{fig:exper_1} compares the throughput of OptChain against the theoretical optimum and the SOTA baselines (Manifoldchain and Prism) across varying error thresholds. The throughput of OptChain closely tracks the theoretical optimum, increasing as the error threshold rises. In contrast, the Prism and Manifoldchain perform significantly worse. As detailed in the proof sketch in Section~\ref{OptChain_tps}, the total error probability stems from three sources: adversarial majority, the absence of honest nodes in a shard, and chain property violations. The second source is the primary contributor, as the other two become negligible by increasing the confirmation depth $k$. Consequently, a higher error threshold enables OptChain, Manifoldchain, and the theoretical optimum to support more shards, boosting parallelism and throughput. Prism, being a non-sharding protocol, maintains nearly constant throughput regardless of the error threshold, appearing as a single point in the figure. OptChain achieves near-optimal performance by achieving both optimal vertical and horizontal scalability, as corroborated by the experimental results presented in Section~\ref{exper_3} and~\ref{exper_2}. As the error threshold increases, the gap between OptChain and the theoretical optimum widens. This divergence occurs because a higher error threshold enables the system to support more shards, which in turn increases the proportion of foreign transaction blocks. When nodes expend bandwidth downloading symbols for these foreign blocks, this consumption does not contribute to the throughput of their own shard. Therefore, the actual throughput increasingly deviates from the theoretical optimum. However, this deviation is mitigated as more nodes join the protocol, as the burden of storing symbols is distributed across a larger number of participants. Additionally, at low error thresholds, Prism outperforms Manifoldchain because Manifoldchain is limited by Bitcoin-like vertical scalability within each shard. However, as the error threshold increases, Manifoldchain surpasses Prism by leveraging its horizontal scalability, a capability that Prism lacks.

\subsection{Vertical Scalability}\label{exper_3}

This experiment evaluates the vertical scalability of OptChain, assessing the system's ability to leverage increased network bandwidth. We fix the network size at 64 nodes. For the sharding protocols (OptChain and Manifoldchain), we set the shard count to 4. Using the \texttt{tc} command, we configure the downlink bandwidth for all nodes, testing discrete values from the set $\{10, 30, 50, 70, 90\}$ Mbps. For each bandwidth configuration, we optimize the parameters of both OptChain and the baselines to maximize their respective throughput.

\begin{figure}
  \centering
  \includegraphics[width=8cm]{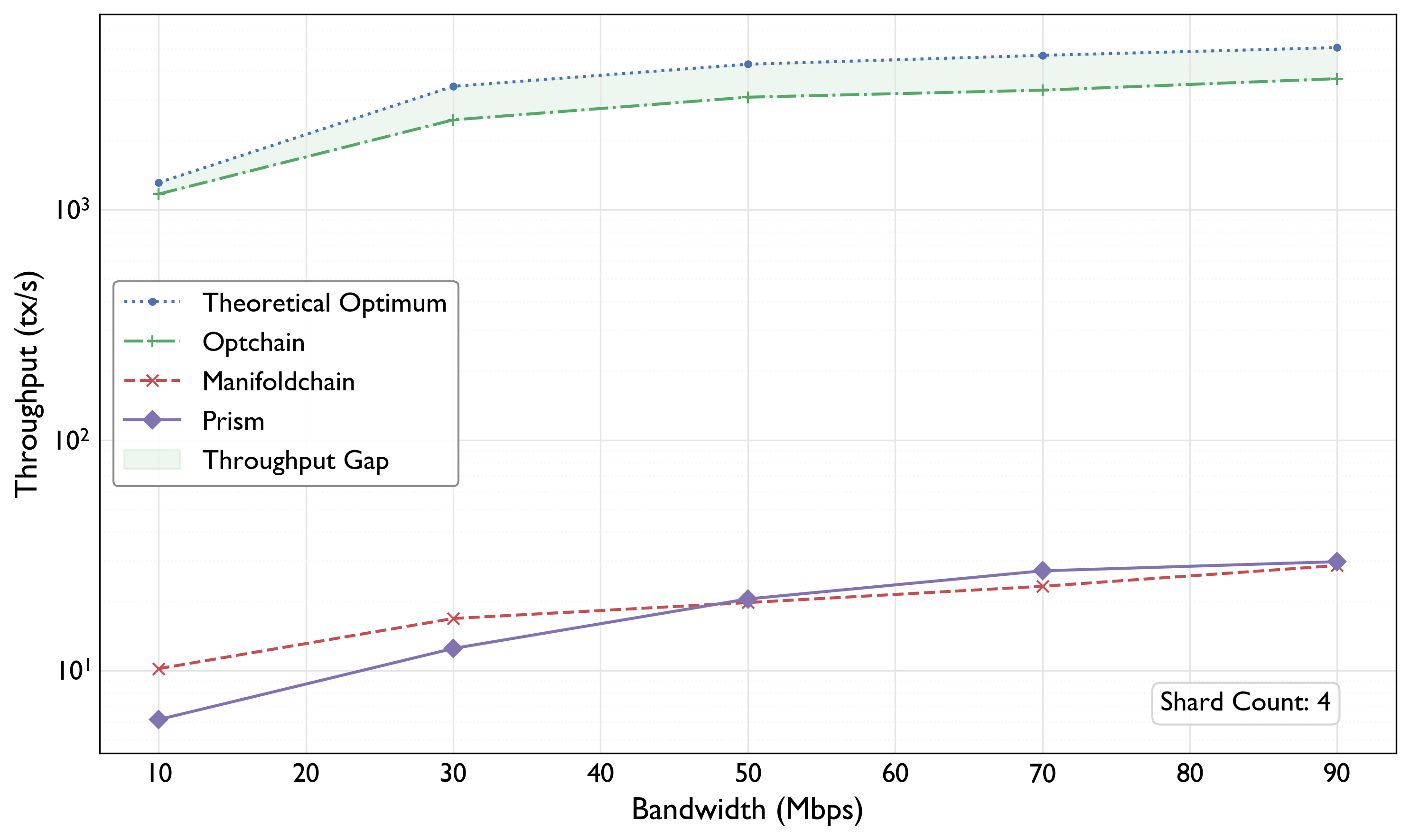}
  \vspace{-2mm}
  \caption{Vertical scalability analysis of OptChain under increasing network bandwidth. OptChain's throughput closely tracks the theoretical optimum and significantly outperforms other baselines. While Prism initially lags behind Manifoldchain due to its lack of sharding, it eventually surpasses Manifoldchain at higher bandwidths due to better vertical scalability.}
  \label{fig:exper_3}
  \vspace{-6mm}\hspace{-10mm}
\end{figure}

Fig.~\ref{fig:exper_3} evaluates the impact of node bandwidth on throughput. Consistent with previous results, OptChain exhibits performance close to the theoretical optimum, while the other baselines perform significantly worse. OptChain achieves this near-optimal vertical scalability by enabling nodes to reach consensus on compact transaction block digests. This design allows the fastest nodes to fully utilize their high bandwidth for downloading transaction blocks without being throttled by slower nodes, matching the throughput upper bound that any consensus protocol can achieve within a single shard. Additionally, at lower bandwidths, Prism achieves lower throughput than Manifoldchain; this is because Prism is non-sharding, whereas Manifoldchain operates with $4$ shards. However, as bandwidth increases, Prism overtakes Manifoldchain because it possesses better vertical scalability than Manifoldchain, whose intra-shard consensus is limited by Bitcoin-like vertical scalability.

\subsection{Horizontal Scalability}\label{exper_2}
This experiment evaluates the horizontal scalability of OptChain, assessing the system's performance as the number of nodes grows. We fix the maximum allowable error probability at $e^{-6}$ and incrementally increase the number of participating nodes. Consistent with our previous methodology, each node is assigned a downlink bandwidth sampled from the real-world dataset. For each network size configuration, we optimize the parameters of OptChain and the remaining baselines to maximize their respective throughput. We exclude Prism from this evaluation, as it is a non-sharding protocol and inherently lacks horizontal scalability.

\begin{figure}
  \centering
  \includegraphics[width=8cm]{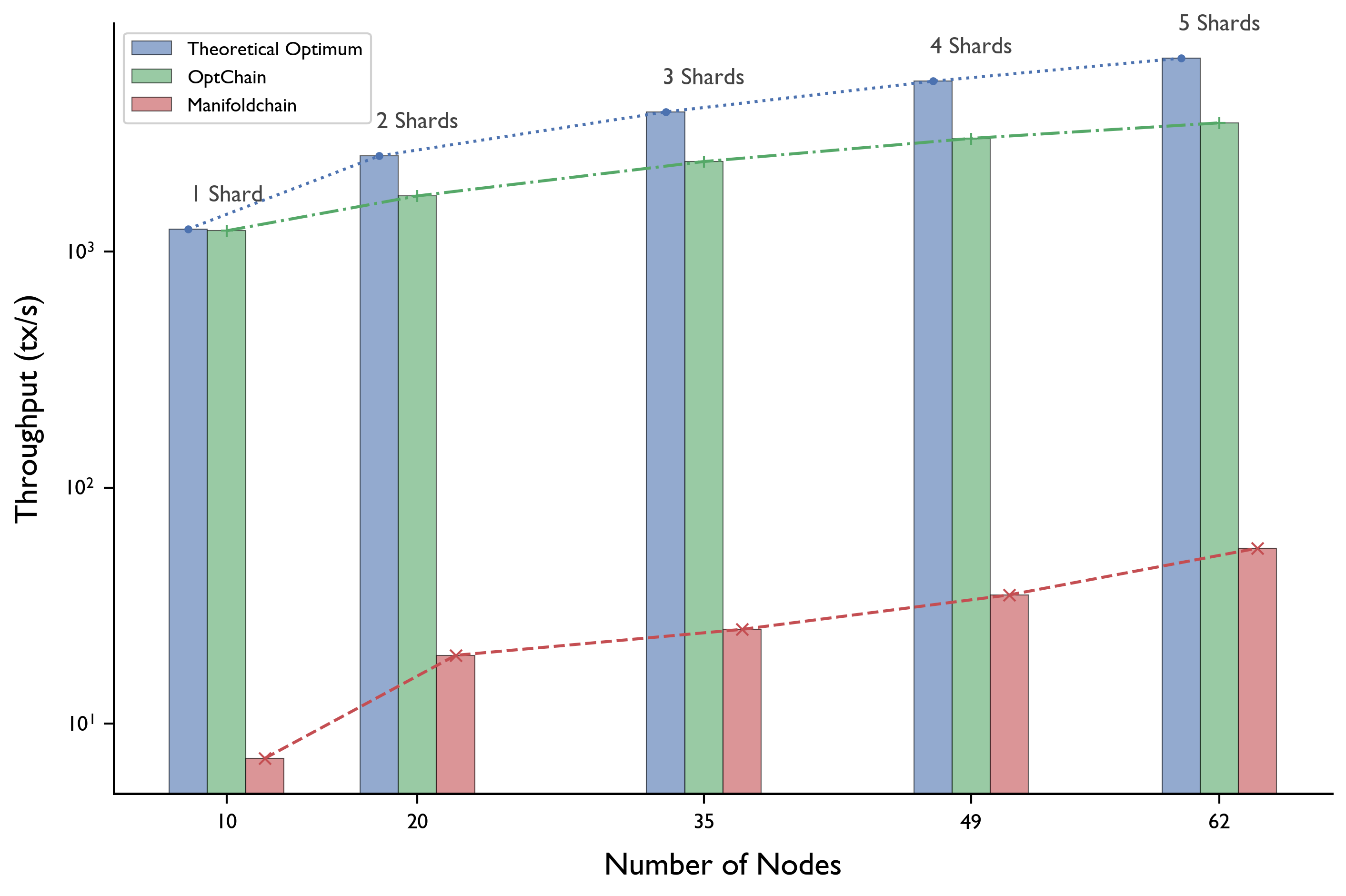}
  \vspace{-2mm}
  \caption{Horizontal scalability evaluation of OptChain with an increasing number of nodes. OptChain closely tracks the theoretical optimum as the node count rises, while Manifoldchain exhibits significantly slower growth. Prism is excluded as it lacks horizontal scalability.}
  \label{fig:exper_2}
  \vspace{-6mm}\hspace{-10mm}
\end{figure}

Fig.~\ref{fig:exper_2} evaluates the horizontal scalability of the protocols by varying the number of nodes. OptChain's throughput increases linearly as the network size grows, closely tracking the theoretical optimum and exhibiting optimal horizontal scalability. In contrast, while Manifoldchain also demonstrates horizontal scalability, its performance remains significantly lower. OptChain achieves this optimality by leveraging diffusion mining to support as many shards as possible, thereby maximizing parallelism. Furthermore, because OptChain also maintains optimal vertical scalability, the incremental throughput contributed by new nodes is determined by the bandwidth of the fastest nodes. Conversely, with the same number of new nodes joining, Manifoldchain's incremental throughput is throttled by the slowest node's bandwidth.

\section{Conclusion}

In this paper, we addressed the fundamental throughput-security trade-off that constrains existing permissionless SMR protocols. We established a theoretical upper bound for throughput in heterogeneous networks and identified two critical barriers to achieving optimality: the limitation of throughput by the slowest nodes (vertical scalability) and the constraints on the number of shards required for security (horizontal scalability).

To overcome these challenges, we introduced OptChain, the first layer-1 permissionless sharding protocol designed to approach this theoretical optimum. OptChain leverages two key innovations: shardis, a novel permissionless verifiable information dispersal scheme that decouples consensus from data availability to enable vertical scaling consistent with the fastest nodes; and diffusion mining, a mechanism that maximizes horizontal scalability by aggregating honest hashing power across shards, ensuring security even with one honest presence per shard. Furthermore, we solved the cross-shard atomicity challenges introduced by our architecture through a global ordering chain and fraud proofs.

We formally proved the security and throughput optimality of our design. Our extensive evaluation, deploying a prototype implementation on geo-distributed AWS EC2 nodes with real-world bandwidth traces, confirmed that OptChain significantly outperforms state-of-the-art protocols like Manifoldchain and Prism. Most importantly, our experimental results demonstrate that OptChain closely approaches the theoretical throughput bound in realistic settings, validating it as a promising direction for next-generation scalable blockchain systems.



%



\printbibliography

\appendices

\section{Artifacts}

To facilitate the double-blind review process and ensure reproducibility of the results presented in this paper, all artifacts—including the OptChain source code, deployment scripts, and configuration files—have been anonymized and uploaded to a persistent repository. The artifact package allows reviewers to compile the OptChain client, run a local minimal working demo, and deploy the full experimental testbed on AWS. The package consists of the following components:

\begin{enumerate}
    \item \textbf{OptChain Core (Rust)~\cite{optchain_anon}:} The source code for OptChain, including the logic for the proposer, availability, and ordering chains, as well as networking and consensus modules. We provide detailed instructions for compiling and running an OptChain node locally. Furthermore, we include documentation on configuring protocol parameters and a quick-start script to launch a simple demonstration.
    \item \textbf{Cloud Deployment (Python/Bash)~\cite{optchain_aws}:} We build a Docker image based on our implementation to facilitate seamless deployment on remote servers, abstracting away system heterogeneity. We also provide a suite of scripts to initialize the environment and deploy the Docker container across multiple AWS EC2 instances simultaneously, automating Docker installation, bandwidth throttling, and experiment orchestration.
\end{enumerate}

\section{Proof of Theorem~\ref{throughput_security_tradeoff}}\label{optimal_throughput_security_trade-off}

We begin by showing how to measure the throughput of a SMR protocol.

\begin{definition}[Throughput]
An SMR protocol $\Pi$ achieves a throughput of $T$ if, with negligible error probability, there exists a log maintained by some honest node such that:
\begin{itemize}
    \item For any time $t$, let the log length be $l$ at time $t$ and $l'$ at time $t+1$; then $l' - l \geq T$.
    \item No violations of safety and liveness.
\end{itemize}
\end{definition}

We prove the theorem as follows:

\begin{proof}
Suppose there exists an SMR protocol that confirms $T$ transactions per second with an error probability $\sigma < \sigma^*(T)$. Assume that the entire network can download at most $x$ transactions per second. Supposed the $i$-th transaction appears in $z_i$ copies across the network, where $\sum_{i=1}^{T} z_i = x$. On average, each transaction appears in at most $\frac{x}{T}$ copies across the network, with each copy downloaded by one node.

The corrupted nodes strictly follow the protocol, honest nodes cannot distinguish them from truly honest ones. Let $X$ be the random variable denoting the node that downloads a specific copy of a transaction, i.e., $X \in \{1, 2, \dots, n\}$. Given that downloads are uniformly distributed, we have:

\begin{equation}
    \Pr(X = j) = \frac{1}{n}, \quad \text{for all } j \in \{1, 2, \dots, n\}.
\end{equation}

Now, let $Y$ be the random variable indicating whether node $j$ downloads the $i$-th transaction. Since there are at most $z_i$ copies of each transaction, each node has a probability of:

\begin{equation}
\begin{split}
\Pr(Y = j) =& 1-\frac{n-1}{n}\frac{n-2}{n-1}\dots\frac{n-z_i}{n-z_i+1} \\
    =& \frac{z_i}{n}, \\
    \quad \text{for all }& i \in \{1, 2, \dots, T\}, j \in \{1, 2, \dots, n\}.
\end{split}
\end{equation}

In this case, availability property is violated if all copies of a transaction are downloaded only by corrupted nodes. The probability of this event (i.e., that none of the $\alpha n$ honest nodes download any copy) is at least:

\begin{equation}
    \sigma = \left(1 - \frac{\min{z_i}}{n}\right)^{\alpha n}.
\end{equation}

The maximum value of $\min{z_i}$ is achieved when every transaction has the same number of copies, so $\max\{\min z_i\} = \frac{x}{T}$. Moreover, since $x \leq n\overline{C}$, it follows that:

\begin{equation}
    \sigma \geq \left(1 - \frac{\overline{C}}{T}\right)^{\alpha n} = \sigma^*(T),
\end{equation}

which contradicts the assumption that $\sigma < \sigma^*(T)$. This completes the proof.
\end{proof}


\section{Security of OptChain}\label{security_proof}

The ultimate goal is to prove that OptChain satisfies the safety and liveness properties defined in Section~\ref{sharding_smr}. We establish this in two steps: (1) we prove that all chains (including the proposer chain, availability chains, and ordering chain) satisfy the CP, CQ, and CG properties, ensuring that honest nodes agree on an ordered log of block commitments; and (2) we prove that OptChain satisfies the security properties of the VID scheme defined in Definition~\ref{VID_security_properties}, ensuring that the log can be effectively converted into a log of available transactions. The outline of the proof is as follows: we present the formal definitions of the CP, CQ, and CG properties in Section~\ref{pow_properties}, and develop the first and second steps in Sections~\ref{proof_step_1} and~\ref{security_PVID}, respectively. Finally, we conclude the proof of Theorem~\ref{OptChain_security_main_text} in Section~\ref{final_proof}.

\subsection{Security Properties}\label{pow_properties}
We focus on a single Bitcoin blocktree and adopt the following notation from~\cite{Bitcoinbackbone}: if \(\mathcal{C}\) is a chain of blocks, then \(\mathcal{C}^{\lceil k}\) denotes the \(k\)-deep prefix of \(\mathcal{C}\), i.e., the chain obtained by removing the last \(k\) blocks from \(\mathcal{C}\). Additionally, given two chains \(\mathcal{C}\) and \(\mathcal{C}'\), we write \(\mathcal{C} \preceq \mathcal{C}'\) if \(\mathcal{C}\) is a prefix of \(\mathcal{C}'\).

\begin{definition}[Common-prefix Property]
The \(k\)-deep common prefix property holds if the \(k\)-deep prefix of the current longest chain remains a prefix of any longest chain at all future times.
\end{definition}

\begin{definition}[Chain Quality Property]
The \((\mu, k)\)-chain quality property holds if, among the last \(k\) consecutive blocks on the longest chain \(\mathcal{C}\), at most a \(\mu\) fraction were mined by the adversary.
\end{definition}

\begin{definition}[Chain Growth Property]
The chain growth property with parameters \(\phi\) and \(s\) states that, over any interval of \(s\) rounds, at least \(\phi s\) blocks are added to the main chain.
\end{definition}
\subsection{CP, CQ, CG of OptChain}\label{proof_step_1}

First, we leverage the proven security of Bitcoin~\cite{nakamotowin} to show that the proposer chain satisfies these properties, as it functions as a standard Bitcoin chain. Second, we demonstrate that OptChain can be reduced to Manifoldchain except with negligible probability and derive the properties of all availability chains based on its results. Finally, using the established properties of the availability chains, we again rely on Bitcoin's proven security to show that the ordering chain also satisfies these properties.

\subsubsection{Security Properties of Proposer Chain}

The proposer chain functions as a standard Bitcoin chain; thus, we can directly derive the following result from prior works~\cite{Bitcoinmadsimple, nakamotowin, Manifoldchain}:

\begin{theorem}\label{proposer_security}
The proposer chain satisfies the CQ, CP, and CG properties provided that 
\[
p_p = \alpha e^{-\lambda_p} > \frac{1}{2}.
\]
Furthermore, its chain growth parameter $g_p$ and chain quality parameter $q_p$ satisfy
\begin{equation}
\begin{split}
g_p &= (1-\delta)\frac{p_p\lambda_p}{1+p_p\lambda_p\underline{\Delta}}, \\
q_p &= 1 - (1+\delta)\frac{1+p_p(1-\alpha)\underline{\Delta}\lambda_p}{p_p},
\end{split}
\end{equation}
except with probability $\mathrm{negl}(\delta)$.
\end{theorem}
\subsubsection{Security Properties of Availability Chains}\label{app_availability_security}

The availability chains can be reduced to Manifoldchain with the following variations:
\begin{itemize}
    \item A malicious block producer may release only a subset of symbols to deceive some honest nodes into voting for an unavailable block, resulting in a reduced effective honest ratio $\alpha' < \alpha$.
    \item Availability blocks do not carry transactions, leading to a lower bounded network delay $\underline{\Delta}$ compared to the typical delay $\overline{\Delta}$ in Manifoldchain.
\end{itemize}

The following theorem shows that our protocol reduces the fraction of deceived honest nodes to a negligible level, yielding $\alpha' \approx \alpha$. Using the updated parameters $\alpha'$ and $\underline{\Delta}$, we derive the properties of the availability chains based on Manifoldchain's results.

\begin{theorem}[Block Availability Consistency]\label{theorem_availability_consistency}
Given a $(vd, d)$-erasure code used in the CMT and a confirmation depth $k$ in the proposer chain, the probability that a node receives all requested symbols of any unavailable transaction block is bounded by
\begin{equation}
\sigma = e\cdot\left(\frac{1}{v}\right)^{k} \cdot \frac{1}{v-1},
\end{equation}
in the limit as $d \rightarrow \infty$.
\end{theorem}

This probability corresponds to the fraction of deceived nodes among honest nodes. It decreases exponentially with $k$, implying that for any $n$ and $f=\beta n$, even if the deceived honest nodes are treated as corrupted, there exists a sufficiently large $k$ such that the fraction of undeceived honest nodes remains above $1/2$. That is, the following condition holds:
$n - f - \sigma \geq f + \sigma$, which simplifies to $n \geq 2f + 2\sigma$.

\begin{proof}
We consider an unavailable block with commitment $\mathtt{com}$, where at most $d-1$ symbols are released to honest nodes. Suppose an honest node requests $i$ symbols of the block. It considers the block available if and only if all of its requested symbols are among the $d-1$ released symbols. The probability of this event, denoted by $\gamma(i)$, satisfies:

\begin{equation}
\begin{split}
\gamma(i) = \frac{d}{dv}\frac{d-1}{dv-1}\dots \frac{d-i+1}{dv-i+1}.
\end{split}
\end{equation}

We have the following recurrence relation for the probability $\gamma(i)$:
\begin{equation}
\frac{\gamma(i)}{\gamma(i-1)} = \frac{d - i + 1}{dv - i + 1} < \frac{d + 1}{dv + 1}.
\end{equation}

Therefore, by induction, we obtain the following upper bound:
\begin{equation}
\gamma(i) < \frac{1}{v} \left( \frac{d + 1}{dv + 1} \right)^i.
\end{equation}

We say that an honest node is deceived if it considers at least one unavailable block to be available. For the latest confirmed proposer block (followed by exactly $k$ blocks), the node requests $k$ symbols. For the second-latest confirmed proposer block (followed by $k+1$ blocks), it requests $k+1$ symbols, and so on. Therefore, the probability that the node is deceived is given by:

\begin{equation}
\begin{split}
\sum_{i=k}^{\infty}\gamma(i) <& \sum_{i=k}^{\infty}\frac{1}{v}\left(\frac{d+1}{dv+1}\right)^i \\
=& \frac{1}{v}\left(\frac{d+1}{dv+1}\right)^k\sum_{i=0}^{\infty} \left(\frac{d+1}{dv+1}\right)^i.
\end{split}
\end{equation}

$\sum_{i=0}^{\infty} \left(\frac{d+1}{dv+1}\right)^i$ is a geometric series $\sum_{k=0}^{\infty}ar^k$, and it converges to $\frac{a}{1-r}$ when $|r|<1$, so

\begin{equation}
\begin{split}
\sum_{i=k}^{\infty}\gamma(i) <& \frac{1}{v}\left(\frac{d+1}{dv+1}\right)^k\sum_{i=0}^{\infty} \left(\frac{d+1}{dv+1}\right)^i \\
=& \left(\frac{d+1}{dv+1}\right)^{k-1}\frac{d+1}{dv(v-1)}.
\end{split}
\end{equation}

When $d\rightarrow \infty$, it approaches

\begin{equation}
\begin{split}
\sum_{i=k}^{\infty}\gamma(i) < \left(\frac{1}{v}\right)^{k}\frac{1}{v-1}.
\end{split}
\end{equation}

In the context of OptChain, each node requests $k$ symbols for each of the $e$ block commitments in a proposer block. Therefore, the final probability is given by:

\begin{equation}
\begin{split}
\gamma = e\sum_{i=k}^{\infty}\gamma(i) < e\left(\frac{1}{v}\right)^{k}\frac{1}{v-1}.
\end{split}
\end{equation}

\end{proof}

To reduce OptChain to Manifoldchain, we introduce the following modifications:
\begin{enumerate}
    \item \textbf{Effective honest ratio:} all deceived honest nodes are treated as corrupted nodes, while the remaining honest nodes are referred to as \textit{effective honest nodes}, with a ratio of $\alpha' = \alpha(1 - \gamma)$.
    \item \textbf{Network delay of availability blocks:} the network delay of an availability block is replaced by $\underline{\Delta}$, which satisfies:
    \begin{itemize}
        \item If an effective honest node receives and accepts an availability block at time $t$, then all effective honest nodes will receive and accept that block by $t + \underline{\Delta}$.
        \item If an effective honest node rejects an availability block, then all effective honest nodes will also reject that block.
    \end{itemize}
\end{enumerate}

Subsequently, we present the following theorem based on Manifoldchain's security:
\begin{theorem}\label{availability_security}
All the availability chains satisfy CP, CQ, and CG properties as long as 
\[
p_a = \frac{\lambda_s \underline{\alpha} + S\lambda_i\alpha'}{\lambda_s+S\lambda_i} e^{-(\lambda_s+\lambda_i)\underline{\Delta}} > \frac{1}{2},
\]    
where $\underline{\alpha} = \alpha'-\log^2(\alpha')$. Furthermore, the chain growth parameter $g_a$ and chain quality parameter $q_a$ satisfy

\begin{equation}
\begin{split}
g_a &= (1-\delta)\frac{p_a(\lambda_s + \lambda_i)}{1+p_a(\lambda_s + \lambda_i)\underline{\Delta}}, \\
q_a &= 1 - (1+\delta)\frac{1+p_a(1-\underline{\alpha})\underline{\Delta}(\lambda_s + \lambda_i)}{p_a},
\end{split}
\end{equation}
except with probability $\mathrm{negl}(\delta)$.
\end{theorem}

\subsubsection{Security Properties of Ordering Chain}

An honest node accepts an ordering block only if it includes references to confirmed availability blocks that all honest nodes have already received and accepted locally. Therefore, the network delay for an ordering block depends solely on its size, which is bounded by $\underline{\Delta}$. Similar to Theorem~\ref{proposer_security}, we can derive the security properties of the ordering chain by applying the results of prior work~\cite{nakamotowin}:

\begin{theorem}\label{ordering_security}
The ordering chain satisfies the CQ, CP, and CG properties provided that 
\[
p_o = \alpha e^{-\lambda_o} > \frac{1}{2}.
\]
Furthermore, its chain growth parameter $g_o$ and chain quality parameter $q_o$ satisfy
\begin{equation}
\begin{split}
g_o &= (1-\delta)\frac{p_o\lambda_o}{1+p_o\lambda_o\underline{\Delta}}, \\
q_o &= 1 - (1+\delta)\frac{1+p_o(1-\alpha)\underline{\Delta}\lambda_o}{p_o},
\end{split}
\end{equation}
except with probability $\mathrm{negl}(\delta)$.
\end{theorem}

\subsection{VID Properties of OptChain}\label{security_PVID}
We present the followint theorem and prove it subsequently:
\begin{theorem}\label{VID_property_theorem}
If the CP, CQ, and CG properties hold for the proposer chain, ordering chain, and all availability chains, then OptChain satisfies all the properties of a VID scheme.
\end{theorem}
\begin{proof}
We prove each property separately.
\begin{itemize}
    \item \textbf{(Termination.)} After invoking \textsc{Disperse}($B_t$), the block proposer broadcasts its $\mathtt{com}$ to all nodes. By the CQ and CG properties of the proposer chain, it will eventually be confirmed, returning \textsc{Disperse}($B_t$)$= \mathtt{com}$. Furthermore, by the CP properties of all availability chains and the ordering chain, every honest node will eventually confirm it in the ordering chain and output $\textsc{Verify(com)} = 1$.
    \item \textbf{(Agreement.)} According to the CP property of the ordering chain, if any honest node confirms a $\mathtt{com}$ (outputs $\textsc{Verify(com)} = 1$), then all honest nodes will do the same by $t + \underline{\Delta}$.
    \item \textbf{(Retrievability.)} If an honest node outputs \textsc{Verify}($\mathtt{com}$)~$=1$ at time $t$, then by agreement, every honest node will output \textsc{Verify}($\mathtt{com}$)~$=1$ by time $t + \Delta$. This implies that each honest node has received $u$ symbols of the block, resulting in a total of $(n - f)u$ randomly sampled symbols across all honest nodes. Among these, at least $u' = (n - f)u - \log^2((n - f)u)$ are distinct with high probability. By choosing $u$ such that $u' > d$, a sufficient number of symbols is available to decode the CMT and reconstruct a block $b'$.
    \item \textbf{(Correctness.)} By the soundness of the CMT, any two honest nodes that output \textsc{Verify}($\mathtt{com}$)~$=1$ will retrieve the same block. If $\mathtt{com}$ was generated by an honest node via \textsc{Disperse}($b$), then $\mathtt{com}$ is the root of the CMT constructed from $b$. Through decoding, every honest node will reconstruct a block $b'$ that is identical to $b$.
\end{itemize}
\end{proof}

\subsection{Proof of Theorem~\ref{OptChain_security_main_text} and Theorem~\ref{OptChain_security_main_text_2}}\label{final_proof}

We first prove the following lemma:

\begin{lemma}\label{convert_chain_q_g_to_u}
Assume the honest majority condition. Let $g$ and $q$ be the chain growth rate and chain quality proportion for a chain. Any valid transaction $\mathrm{tx}$ submitted at $t$ will be confirmed by $t+u$, where $u = \frac{k + \lceil 1/q \rceil}{g}$.
\end{lemma}

\begin{proof}
During the time interval $[t, t+u]$, the chain grows by at least $g \cdot u = k + \lceil 1/q \rceil$ blocks, according to the chain growth property. Of these newly generated blocks, the most recent $k$ blocks serve as the persistence buffer. This pushes at least $\ell = \lceil 1/q \rceil$ newly mined blocks deep enough into the chain to become part of the stable prefix (i.e., buried by at least $k$ blocks), making them confirmed. 

By the chain quality property, any sequence of $\ell$ blocks is guaranteed to contain at least $q \cdot \ell \ge 1$ honest block. Because honest miners strictly follow the protocol by including valid pending transactions, and adversaries cannot forge conflicting transactions to invalidate $\mathrm{tx}$, this honest block must contain the transaction $\mathrm{tx}$ submitted at $t$. Therefore, $\mathrm{tx}$ is permanently recorded and confirmed by time $t+u$.
\end{proof}

Now we prove OptChain's safety and liveness.
\begin{proof}
Combining Theorems~\ref{proposer_security}, \ref{availability_security}, \ref{ordering_security}, and~\ref{VID_property_theorem}, we show that OptChain satisfies both safety and liveness:

\noindent {\bf Safety:} Honest nodes retrieve and execute transaction blocks belonging to their shards once their commitments are confirmed in the ordering chain. Thus, all honest nodes in the same shard retrieve the same transaction blocks confirmed at identical positions. If two honest nodes in the same shard were to execute different transaction sequences, the correctness of the VID scheme would imply different commitments confirmed at the same position in the ordering chain, contradicting the CP property.

\noindent \textbf{Liveness:} If an honest client submits a transaction $\mathtt{tx}$ at time $t$, all honest nodes receive it by $t+\underline{\Delta}$. An honest node will include $\mathtt{tx}$ in a transaction block $B_t$ if it belongs to its shard; otherwise, it forwards $\mathtt{tx}$ to the appropriate shard for packaging. By Lemma~\ref{convert_chain_q_g_to_u}, $B_t$ is confirmed in the proposer chain by $t+\underline{\Delta} + \frac{k+\lceil 1/q_p \rceil}{g_p}$. Accounting for network delay, $B_t$ is confirmed by all honest nodes by $t + \frac{k+\lceil 1/q_p \rceil}{g_p} + 2\underline{\Delta}$. All honest nodes then request the corresponding symbols, receiving them optimistically by $t + \frac{k+\lceil 1/q_p \rceil}{g_p} + 4\underline{\Delta}$. Similarly, by Lemma~\ref{convert_chain_q_g_to_u}, $B_t$ is confirmed in the availability chain by $t + \frac{k+\lceil 1/q_p \rceil}{g_p} + \frac{k+\lceil 1/q_a \rceil}{g_a} + 4\underline{\Delta}$. This confirmation becomes consistent across all honest nodes by $t + \frac{k+\lceil 1/q_p \rceil}{g_p} + \frac{k+\lceil 1/q_a \rceil}{g_a} + 5\underline{\Delta}$. Subsequently, honest nodes mine ordering blocks that reference confirmed availability blocks. As a result, the availability block containing $B_t$ is confirmed by $t + \sum_{x\in\{p, a, o\}}\frac{k+\lceil 1/q_x \rceil}{g_x} + 5\underline{\Delta}$. Therefore, any transaction submitted at time $t$ is fully confirmed by $t + u$, where $u = \sum_{x\in\{p, a, o\}}\frac{k+\lceil 1/q_x \rceil}{g_x} + 5\underline{\Delta}$.


\end{proof}


\section{Cross-shard Transaction Atomicity}\label{cross_tx_atomicity_proof}

\begin{proof}
We sequentially verify that each condition of cross-shard atomicity outlined in Definition~\ref{cross_shard_atomicity_definition} is satisfied.
\begin{itemize}
    \item If all withdrawal-txs are confirmed, honest nodes will accept every valid deposit-tx. By the liveness property of OptChain, all deposit-txs will eventually be confirmed.
    \item If any deposit-tx is confirmed, the safety property of OptChain ensures that all honest nodes have accepted it; this implies that all corresponding withdrawal-txs appear before the deposit-tx on the ordering chain. Since the deposit-tx is confirmed, all preceding withdrawal-txs (which are at a greater depth) are implicitly confirmed. Following the logic of the first case, all remaining deposit-txs will eventually be confirmed.
    \item If the cross-tx is initiated by an honest user, the liveness of OptChain guarantees that all withdrawal-txs are eventually confirmed. Consequently, as established above, all deposit-txs will eventually be confirmed.
\end{itemize}
\end{proof}

\section{OptChain's Throughput}\label{proof_opt_tps}

OptChain requires each shard to contain at least one honest node to ensure security. We first present a lemma used to calculate the probability that this condition is violated, which serves as the primary source of error probability in our protocol.

\begin{lemma}\label{probability_of_honest_existence_lemma}
    Given $n$ honest nodes randomly distributed into $m$ shards, for any shard, the probability $\varepsilon(n)$ there is no honest nodes is bounded by $(1- \frac{1}{m})^{n}$.
\end{lemma}

\begin{proof}
    The distribution of $n$ nodes into $m$ shards satisfies a Multinomial Distribution\cite{multinomial_distribution}. Specifically, it models the probability of counts for each side of a $m$-sided dice rolled $n$ times. For $n$ independent trials each of which leads to a success for exactly one of $m$ categories, with each category having a given fixed success probability, the multinomial distribution gives the probability of any particular combination of numbers of successes for the various categories. Mathematically, for each independent trial, we have $m$ possible mutually exclusive outcomes, with corresponding probabilities $p_0, ..., p_{m-1}$. Given a random distribution, $p_0=p_1=...=p_{m-1}=\frac{1}{m}=p$. The probability mass function of this multinomial distribution is:
    \begin{equation}
        \begin{split}
            f(x_0,...,x_{m-1};&n,p) = Pr[X_0=x_0, ..., X_m=x_m] \\
            &=\frac{n!}{x_0!\cdot ...\cdot x_{m-1}!}p^{x_0}\times ... \times p^{x_{m-1}},\\
            &=\frac{n!}{x_0!\cdot ...\cdot x_{m-1}!}p^{n},
        \end{split}
    \end{equation}where $\sum_{i=0}^{m-1}x_i = n$. The probability that there is no one honest node in shard $i$ is denoted as follows:

    \begin{equation}
    \begin{split}
    \sum_{\mathbf{x}\setminus\{x_i\}}Pr[x_0 &\geq 0,\dots x_i=0, ..., x_{m-1} \geq 0] \\
    =& \sum_{\mathbf{x}\setminus\{x_i\}}\frac{n!}{x_0!\cdot...\cdot 0!\cdot...\cdot x_{m-1}!}p^n \\
    =& \frac{p^n}{{p'}^n}\sum_{\mathbf{x}\setminus\{x_i\}}\frac{n!}{x_0!\cdot...\cdot x_{m-2}!}{p'}^n,
    \end{split}
    \end{equation}where $p' = \frac{1}{m-1}$, $\frac{n!}{x_0!\cdot ... \cdot x_{m-2}}{p'}^{n}$ represents the probability mass function of another multinomial distribution where $n$ nodes are randomly distributed into $m-1$ shards. Given 
    \begin{equation}
        \begin{split}
            \sum_{\mathbf{x}\setminus\{x_i\}}\frac{n!}{x_0!\cdot ... \cdot x_{m-2}}{p'}^{n} = 1,
        \end{split}
    \end{equation}we have

    \begin{equation}
    \begin{split}
    \sum_{\mathbf{x}\setminus\{x_i\}}Pr[x_0 &\geq 0,\dots x_i=0, ..., x_{n-1} \geq 0] \\
    =& \left(\frac{\frac{1}{m}}{\frac{1}{m-1}}\right)^n = \left(1-\frac{1}{m}\right)^n.
    \end{split}
    \end{equation}
\end{proof}

We then prove the Theorem~\ref{opt_tps}.

\begin{proof}
Following the parameter setting in Section~\ref{OptChain_tps}, we first calculate the error probability that, for a transaction proposed by an honest node, at least one security property is violated. A violation occurs if any of the following events happens:
\begin{enumerate}
    \item \textbf{Adversarial majority.} The fraction of corrupted nodes plus deceived honest nodes exceeds $50\%$, with probability $\sigma_1 = \texttt{negl}(k)$.
    \item \textbf{Honest absence.} Given that there are $n - f$ honest nodes and $\frac{T}{\overline{C}}$ shards, the probability that a transaction's shard contains no honest node is $\sigma_2 = \left(1 - \frac{\overline{C}}{T} \right)^{\alpha n}$.
    \item \textbf{Chain properties violation.} One or more of the CP, CQ, or CG properties of the proposer chain or its affiliated availability chain is violated, with probability $\sigma_3 = \texttt{negl}(k)$.
\end{enumerate}

Thus, the total error probability is 
\[
\sigma = \sigma_1 + \sigma_2 + \sigma_3 = \left(1 - \frac{\overline{C}}{T}\right)^{\alpha n} + \texttt{negl}(k).
\]
By backward reasoning, we can derive the optimal throughput $\frac{\overline{C}}{1 - {(\sigma-\texttt{negl}(k))}^{\frac{1}{\alpha n}}}$, which matches the throughput bound stated in Theorem~\ref{opt_tps}.

\end{proof}

\section{Communication Complexity}\label{communication_complexity}

In this section, we analyze the communication cost of our protocol in the case where transaction block reconstruction is not required, as in prior VID works~\cite{VID, DBLP:conf/podc/HendricksGR07, AVID, Dispersedledger}. If reconstruction were required for every node, then the communication cost per node would be information-theoretically $\Omega(b)$. We measure the size of a block body $B_b$ as $O(b)$, and a block header $B_h$ as $O(1)$. We define the \textit{communication complexity} of a VID protocol as the total size of all messages required to complete \textsc{Disperse} and \textsc{Verify} $= 1$ for a transaction block proposed by an honest node. Specifically, we analyze the communication cost of \textsc{Disperse} and \textsc{Verify} separately, as they correspond to the mining of the proposer chain and the availability chain, respectively.

\begin{itemize}
    \item \textbf{Proposer chain (\textsc{Disperse($B_t$)}).} When an honest node mines a transaction block, it broadcasts the block header $B_h$ to all nodes, incurring a communication cost of $n \cdot |B_h|$. Additionally, a proposer block of size $|B_h| + |\mathtt{prop\_com\_set}|$, which includes the $\mathtt{com}$(s) of $e$ transaction blocks, is broadcast to all nodes, adding $n(|B_h| + |\mathtt{prop\_com\_set}|)$ to the total cost. Thus, the amortized communication complexity per transaction block is $n\left[|B_h| + \frac{1}{e}\left(|B_h| + |\mathtt{prop\_com\_set}|\right)\right] = n\left(1 + \frac{1}{e}\right)|B_h| = O(n)$.
    \item \textbf{Availability chain (\textsc{Verify(com)}).} Similar to a proposer block, an availability block of size $|B_h| + |\mathtt{avai\_com\_set}|$, containing the $\mathtt{com}$(s) of $e$ transaction blocks, is broadcast to all nodes, resulting in a communication complexity of $n(|B_h| + |\mathtt{avai\_com\_set}|)$. Hence, the amortized communication cost per transaction block is $n\left[|B_h| + \frac{1}{e}\left(|B_h| + |\mathtt{avai\_com\_set}|\right)\right] = n\left(1 + \frac{1}{e}\right)|B_h| = O(n)$. For symbol sampling, each in-shard node requests $k$ symbols and each out-shard node requests $u$ symbols per block, yielding a total communication overhead of $\left(\frac{n}{S}\right)\cdot \frac{k}{d} \cdot O(b) + \left(1 - \frac{n}{S}\right)\cdot \frac{u}{d} \cdot O(b)$. Given that the coding scheme ensures $\frac{\alpha n}{S}\cdot k = O(d)$ and $\left(1 - \frac{\alpha n}{S}\right)\cdot u = O(\frac{d}{\alpha})$, the overall communication complexity for sampling is $O(b)$, resulting in a total of $O(n+\frac{d}{\alpha})$.
\end{itemize}

In conclusion, the total communication complexity is $O(n + \frac{d}{\alpha})$, and the per-node communication complexity is $O\left(\frac{b}{\alpha n}\right)$, sublinear in the block size.

\end{document}